\documentclass[11pt, a4paper, reqno, 11pt]{article}
\usepackage{a4wide}
\usepackage{microtype}
\usepackage{amsmath}
\usepackage{amssymb}
\usepackage{amsthm} 
\usepackage{hyperref}
\hypersetup{colorlinks=true,citecolor=blue,linkcolor=blue,urlcolor=blue}
\usepackage{float}
\usepackage{enumerate}
\usepackage{thmtools,thm-restate}
\usepackage{mathtools}
\usepackage[babel]{csquotes}
\usepackage{pgf, tikz}
\usetikzlibrary{positioning, shadows, arrows, patterns, trees, shapes, fit}
\usepackage{relsize}
\tikzset{fontscale/.style = {font=\relsize{#1}}
    }

\newtheorem{theorem}{Theorem}
\newtheorem{lemma}[theorem]{Lemma}

\newtheorem{corollary}[theorem]{Corollary}
\theoremstyle{definition}
\newtheorem{definition}[theorem]{Definition}
\newtheorem{example}[theorem]{Example}
\newtheorem{remark}[theorem]{Remark}

\newcommand{\VCSPof}[1]{$\operatorname{VCSP}\hspace*{-0.2em}\left(#1\right)$}
\newcommand{\VCSPs}[1]{$\operatorname{VCSP}_s\hspace*{-0.2em}\left(#1\right)$}
\newcommand{\VCSPl}[1]{$\operatorname{VCSP}_l\hspace*{-0.2em}\left(#1\right)$}
\newcommand{\VCSPx}[2]{$\operatorname{VCSP}_{#1}\hspace*{-0.2em}\left(#2\right)$}
\newcommand{\GMC}{$\operatorname{GMC}$}
\newcommand{\BGMC}[2]{$\operatorname{GMC}_{#1}^{#2}$}
\newcommand{\EDS}{\textsf{EDS}}
\newcommand{\SIM}{\textsf{SIM}}
\newcommand{\SEDS}{\textsf{SEDS}}
\newcommand{\SDS}{\textsf{SDS}}

\DeclarePairedDelimiter\myparentheses{\lparen}{\rparen}
\DeclarePairedDelimiter\mybrackets{\lbrack}{\rbrack}
\newcommand{\fixall}[1]{\operatorname{fix} \myparentheses*{#1}}
\newcommand{\fix}[2]{\operatorname{fix}_{#1} \mybrackets*{#2}}

\newcommand{\eoe}{\hfill$\clubsuit$} 

\begin{document}
\title{Using a min-cut generalisation to go\\ beyond Boolean surjective VCSPs\thanks{An extended abstract of this work
appeared in \emph{Proceedings of the 36th International Symposium on
Theoretical Aspects of Computer Science (STACS 2019)}~\cite{mz19:stacs}. 
Stanislav \v{Z}ivn\'y was supported by a Royal Society University Research Fellowship. The work was done while
Gregor Matl was at the University of Oxford. This project has received funding
from the European Research Council (ERC) under the European Union's Horizon 2020
research and innovation programme (grant agreement No 714532). The paper
reflects only the authors' views and not the views of the ERC or the European
Commission. The European Union is not liable for any use that may be made of the
information contained therein.}}

\author{
Gregor Matl\\
Department of Informatics, Technical University of Munich, Germany\\
\texttt{matlg@in.tum.de}
\and
Stanislav \v{Z}ivn\'y\\
Department of Computer Science, University of Oxford, UK\\
\texttt{standa.zivny@cs.ox.ac.uk}
}

\date{}
\maketitle
 
\begin{abstract}
In this work, we first study a natural generalisation of the Min-Cut problem, where a graph is augmented by a superadditive set function defined on its vertex subsets. The goal is to select a vertex subset such that the weight of the induced cut plus the set function value are minimised. In addition, a lower and upper bound is imposed on the solution size. We present a polynomial-time algorithm for enumerating all near-optimal solutions of this Bounded Generalised Min-Cut problem.

Second, we apply this novel algorithm to surjective general-valued constraint satisfaction problems (VCSPs), i.e., VCSPs in which each label has to be used at least once. On the Boolean domain, Fulla, Uppman, and \v{Z}ivn\'y~[ACM ToCT'18] have recently established a complete classification of surjective VCSPs based on an unbounded version of the Generalised Min-Cut problem. Their result features the discovery of a new non-trivial tractable case called EDS that does not appear in the non-surjective setting.

As our main result, we extend the class EDS to arbitrary finite domains and provide a conditional complexity classification for surjective VCSPs of this type based on a reduction to smaller domains. On three-element domains, this leads to a complete classification of such VCSPs.
\end{abstract}

\section{Introduction}
\label{sec:intro}

Constraint satisfaction problems (CSPs) are fundamental computer science
problems studied in artificial intelligence, logic (as model checking of the
positive primitive fragment of first-order logic), graph theory (as
homomorphisms between relational structures), and databases (as conjunctive
queries)~\cite{HN2008}. A vast generalisation of CSPs is that of general-valued
CSPs (VCSPs)~\cite{Schiex95:valued}, see also~\cite{cohen06:complexitysoft}. Recent years have seen some
remarkable progress on our understanding of the computational complexity of CSPs
and VCSPs, as will be discussed later in related work. We start with a few
definitions to state existing as well as our new results.

We consider regular, surjective as well as lower-bounded VCSPs on the extended
rationals $\overline{\mathbb{Q}} = \mathbb{Q} \cup \left\{ \infty \right\}$. An
\textit{instance} $I = \left(V, D, \phi_I\right)$ of either of these problems is given by a finite set of variables $V = \left\{x_1, \dots, x_n\right\}$, a finite set of labels $D$ called the \textit{domain}, and an \textit{objective function}
$\phi_I : D^n  \rightarrow \overline{\mathbb{Q}}$. The objective function is of the form 
\begin{equation*}
\phi_I\left(x_1,\dots, x_n\right) = \sum_{i = 1}^{t}w_i \cdot\gamma_i\left(\boldsymbol{x}_i\right), 
\end{equation*}
where $t \in \mathbb{N}$ and, for each $1 \leq i \leq t$, $\gamma_i  :
D^{\operatorname{ar}\left( \gamma_i \right)}\rightarrow \overline{\mathbb{Q}}$
is a \textit{weighted relation} of arity $\operatorname{ar}\left( \gamma_i\right)
\in \mathbb{N}$, $w_i \in \mathbb{Q}_{ \geq 0}$ is a \textit{weight} and $\boldsymbol{x}_i \in V^{\operatorname{ar}\left( \gamma_i\right)}$ is a tuple of variables from $V$ called the \textit{scope} of $\gamma_i$.

Regular, surjective and lower-bounded VCSPs differ only in their solution space, although this makes a big difference in complexity. If $I$ is an instance of a regular VCSP, an \textit{assignment} is a map $s : V
\rightarrow D$ assigning a label from $D$ to each variable. In the surjective
setting, only a surjective map $s : V \rightarrow D$ is an assignment.
For lower-bounded VCSPs, a fixed lower bound $l : D \rightarrow \mathbb{N}_0$ is
provided and an assignment is a map $s : V \rightarrow D$ such that $\left|s^{ -
1}\left(d\right)\right| \geq l\left(d\right)$ for every label $d \in D$. In
other words, a lower bound $l\left(d\right)$ on the number of occurrences of
each label $d \in D$ is imposed.
This is a generalisation of surjective VCSPs where the lower bound is always 1.
(We are not aware of any previous work on lower-bounded VCSPs, which we
introduce in this work.)
The \textit{value} of an assignment $s$ is given by
$\phi_I\left(s\left(x_1\right), \dots, s\left(x_n\right)\right)$. An assignment
is called \textit{feasible} if its value is finite, and is called \textit{optimal} if it is of minimal value among all assignments for the instance. The objective is to obtain an optimal assignment. 

While finding an optimal assignment is NP-hard in general, valued constraint
languages impose a natural restriction on the types of instances that are
allowed. A \textit{valued constraint language}, or simply a \textit{language},
is a possibly infinite set of weighted relations. In this paper, we only
consider languages of bounded arity, that is languages admitting a fixed upper
bound on the arity of all weighted relations contained in them. 
Weighted relations in any VCSP instance will be stored explicitly.
 
We denote the class of regular VCSP instances with objective functions using only weighted relations from a language $\Gamma$ by \VCSPof{\Gamma}. Similarly, \VCSPs{\Gamma} is the class of surjective VCSP instances with weighted relations from $\Gamma$ and, for some lower bound $l : D \rightarrow
\mathbb{N}_0$, \VCSPl{\Gamma} is the class of lower-bounded VCSP instances over
$\Gamma$ with bound $l$.

A language $\Gamma$ is \textit{globally tractable} if there is a polynomial-time
algorithm for solving each instance of \VCSPof{\Gamma},
or \textit{globally intractable} if \VCSPof{\Gamma} is NP-hard.
Analogously, $\Gamma$ is \textit{globally $s$-tractable} if there is a polynomial-time
algorithm for \VCSPs{\Gamma}, or \textit{globally $s$-intractable} if
\VCSPs{\Gamma} is NP-hard. And $\Gamma$ is \textit{globally $\ell$-tractable} if
\VCSPl{\Gamma} is solvable in polynomial time for every fixed lower bound $l :
D \rightarrow \mathbb{N}_0$, or \textit{globally $\ell$-intractable} if
\VCSPl{\Gamma} is NP-hard for at least one fixed lower bound $l : D \rightarrow
\mathbb{N}_0$. Thus, global $\ell$-tractability implies global $s$-tractability,
and global $s$-intractability implies global $\ell$-intractability.

The following examples show how well-studied variants of the \textsc{Min-Cut}
problem can be modelled in the VCSP frameworks we have defined.

\begin{example}[\textsc{$r$-Terminal Min-Cut}]
Given a graph $G = \left(V, E\right)$ with non-negative edge weights $w : E
\rightarrow \mathbb{Q}_{ \geq 0}$ and designated terminal vertices $s_1, \dots,
s_r \in V$, the \textsc{$r$-Terminal Min-Cut} problem asks to partition $V$ into
subsets $X_1, \dots, X_r$ such that $s_d \in X_d$ for all $d \in
\left[r\right]:=\{1,\ldots,r\}$, while the accumulated weight of all edges going
between distinct sets $X_i$ and $X_j$ is minimised. For $r = 2$, this problem is
also known as the \textsc{$\left(s, t\right)$-Min-Cut} problem.

We show how this problem can be represented as a regular VCSP. Let $\gamma_{\text{$r$-cut}}$ denote the binary weighted relation defined for $x, y \in
\left[r\right]$ by $\gamma_{\text{$r$-cut}}\left(x, y\right) = 0$ if $x = y$ and
$\gamma_{\text{$r$-cut}}\left(x, y\right) = 1$ otherwise. Furthermore, for each label $d
\in \left[r\right]$, let $\rho_d$ denote the constant relation given by $\rho_d\left(d\right) = 0$ and $\rho_d\left(x\right) =  \infty$ for $d \neq x \in
\left[r\right]$. Let $\Gamma_{\text{$r$-cut}} =
\left\{ \gamma_{\text{$r$-cut}},  \rho_1,  \dots,  \rho_{r}\right\}$.

Finding an optimal $r$-terminal cut is equivalent to solving the
\VCSPof{\Gamma_{\text{$r$-cut}}} instance $I = \left(V, \left[r\right],
\phi\right)$ with objective function \begin{equation*} \phi\left(x_1, \dots,
x_n\right) =  \rho_1\left(s_1\right) +  \dots +  \rho_{r}\left(s_{r}\right) +
\sum_{\left(u, v \right) \in E}w\left(u, v\right) \cdot
\gamma_{\text{$r$-cut}}\left(u, v\right). \end{equation*}
To see this, observe that there is a correspondence between feasible assignments $s : V \rightarrow \left[r\right]$ and $r$-terminal cuts $X_1, \dots, X_r$ by setting  $X_d =
\left\{v \in V : s\left(v\right) = d\right\}$. Hence, an optimal assignment
induces an optimal cut.

The \textsc{$r$-Terminal Min-Cut} problem can be solved in polynomial time if $r
= 2$, but it is NP-hard for any $r \geq 3$~\cite{Dahlhaus:1994}. Since every
\VCSPof{\Gamma_{\text{$r$-cut}}} instance can also be reduced to an instance of the
\textsc{$r$-Terminal Min-Cut} problem, the language $\Gamma_{\text{$r$-cut}}$ is
globally tractable if $r = 2$ and globally intractable for $r \geq 3$. 
\eoe
\end{example} 

\begin{example}[\textsc{$r$-Way Min-Cut}]
Without setting out any terminals, the \textsc{$r$-Way Min-Cut} problem asks to partition $V$ into non-empty subsets $X_1, \dots, X_r$ such that weight of the induced cut is minimised.
Finding an optimal $r$-way min-cut is equivalent to solving the
\VCSPs{\left\{\gamma_{\text{$r$-cut}}\right\}} instance $I = \left(V,
\left[r\right], \phi\right)$ with objective function
\begin{equation*}
	 \phi\left(x_1, \dots, x_r\right) = \sum_{\left(u, v \right) \in E}w\left(u, v\right) \cdot \gamma_{\text{$r$-cut}}\left(u, v\right).
\end{equation*}
The \textsc{$r$-Way Min-Cut} problem can be solved in polynomial time for every
fixed integer $r$~\cite{Goldschmidt:1994}. Since every
\VCSPs{\left\{\gamma_{\text{$r$-cut}}\right\}} instance can be reduced to an \textsc{$r$-Way Min-Cut} problem as well, the language $\left\{\gamma_{\text{$r$-cut}}\right\}$ is globally $s$-tractable.
\eoe
\end{example}

For a fixed $l : D \rightarrow V$, 
\VCSPl{\left\{\gamma_{\text{$r$-cut}}\right\}} allows to model a generalisation
of the \textsc{$r$-Way Min-Cut} problem where a partition $X_1, \dots, X_r$ of
$V$ minimising the induced cut is sought under the condition that
$\left|X_d\right| \geq l\left(d\right)$ for every $d \in D$. As far as we know, the complexity of
both \VCSPl{\left\{\gamma_{\text{$r$-cut}}\right\}} and the lower-bounded
\textsc{$r$-Way Min-Cut} problem is unknown.

\paragraph*{Related Work}

Early results on CSPs include the fundamental results of Schaefer on Boolean
CSPs~\cite{schaefer78:complexity} and of Hell and Ne\v{s}et\v{r}il on graph
CSPs~\cite{Hell:1990}. The computational complexity of CSPs has drawn a lot of
attention following the seminal paper of Feder and Vardi~\cite{Feder:1998}.
Using the algebraic approach~\cite{Jeavons:1997,Bulatov05:classifying}, the
complexity of CSPs on finite domains was resolved in two independent papers by
Bulatov~\cite{Bulatov:2017} and Zhuk~\cite{Zhuk:2017}. The computational
complexity of the problem of minimising the number of unsatisfied constraints
(and more generally rational-valued weighted relations) was obtained by Thapper and \v{Z}ivn\'y in~\cite{Thapper:2016}.
Finally, the computational complexity of general-valued CSPs on finite domains
was obtained by the work of Kozik and Ochremiak~\cite{Kozik:2015} and
Kolmogorov, Krokhin, and Rol\'inek~\cite{Kolmogorov:2017}.

Many constraint solvers allow not only constraints that apply locally to the
variables specified as arguments, but also some sort of global constraints. In
fact, the latter are the default representations in most constraint
solvers~\cite{Rossi06:handbook}. Among VCSPs with global constraints studied
from the complexity point of view are CSPs with global cardinality constraints, or CCSPs, where it is specified how often exactly each label has to occur in an assignment. A dichotomy theorem for CCSPs on finite domains was established by Bulatov and Marx~\cite{Bulatov:2010cardinality}.

Surjective VCSPs, which can be seen as imposing a global constraint, have been
studied by Fulla, Uppman, and \v{Z}ivn\'y~\cite{FUZ18:surjective}, following
earlier results on CSPs by Creignou and H\'ebrard~\cite{Creignou1997} and
Bodirsky, K\'ara, and Martin~\cite{Bodirsky12:dam}. Unfortunately, the algebraic
approach that has proved pivotal in the understanding of the computational
complexity of regular CSPs and VCSPs is not applicable in the surjective
setting.

The following two facts are easy to show (see, e.g,~\cite{FUZ18:surjective}): (i)
intractable languages are also $s$-intractable; (ii) a tractable language $\Gamma$
is also $s$-tractable \emph{if} $\Gamma$ includes all constant relations.
Consequently, new $s$-tractable languages can only occur (if at all) as subsets
of tractable languages that do not contain all constant relations. The first example of such languages have been presented in~\cite{FUZ18:surjective}. In
particular, the authors of \cite{FUZ18:surjective} have identified languages on the Boolean domain that are \textit{essentially a
downset}, or {\EDS}, as a new class of efficiently solvable problems and, in
doing so, have provided a complexity classification of surjective VCSPs on the Boolean domain. 
Informally, a weighted relation $\gamma$ is {\EDS} if both the set of
feasible tuples of $\gamma$ and the set of optimal tuples of $\gamma$ are
essentially downsets. Here a relation is called essentially a downset if it can
be written as a conjunction of downsets and binary equality relations.
Equivalently, a relation is essentially a downset if it admits a binary
polymorphism $\operatorname{sub}(x,y)=\min(x,1-y)$. A finite language is {\EDS} if every
weighted relation in $\Gamma$ is {\EDS}. The definition for infinite languages
is more complicated. We give a formal definition of the {\EDS} class in
Section~\ref{sec:EDS} and refer the reader to~\cite{FUZ18:surjective} for more
details.

The tractability result of {\EDS} languages is based on the Generalised Min-Cut
({\GMC}) problem for graphs, also introduced in~\cite{FUZ18:surjective}. In a
{\GMC} instance, the goal is to find a non-trivial subset of the vertices such
that the weight of the induced cut and a superadditive set function are
minimised simultaneously. In particular, the following has been shown
in~\cite{FUZ18:surjective}. Firstly, the objective function of surjective VCSPs
that are {\EDS} can be approximated by an instance of the {\GMC} problem.
Secondly, there is a polynomial-time algorithm to enumerate all solutions to the
{\GMC} problem that are optimal up to a constant factor. These two together give
an efficient algorithm for surjective VCSPs that are {\EDS}.

\paragraph*{Contributions}

This paper extends the class {\EDS} to arbitrary finite domains. We introduce a
class {\SIM} of languages that exhibit properties \textit{similar to a Boolean language}.
Based on this class, we define the class {\SEDS} of languages \textit{similar to {\EDS}} as a natural extension of
{\EDS} and classify languages from this extension based on two criteria.
Firstly, we give a subclass {\SDS}, or \textit{similar to a downset}, of {\SEDS} that guarantees global
$\ell$-tractability without additional requirements. Secondly, we prove that the
complexity of lower-bounded VCSPs over any remaining {\SEDS} languages is equivalent to the complexity over a particular language on a smaller domain, which can be constructed by including all possible
ways to assign a certain label.
This is illustrated in Figure~\ref{figure_three_element_domain} (left), where we use
the notation $\fixall{\Gamma}$, formally defined in Section~\ref{subsec:fix}.
Informally, for a language $\Gamma$ defined on domain $D$ that includes the
label $0$, $\fixall{\Gamma}$ is the language on domain $D\setminus\{0\}$ obtained
by including, for every
weighted relation $\gamma\in\Gamma$ of arity $n$ and a subset $U$ of
the arguments of $\gamma$, the weighted relation $\fix{U}{\gamma}$, which is a
weighted relation on $D\setminus\{0\}$ of arity $n-|U|$ defined as the
restriction of $\gamma$ that fixes the label $0$ to all arguments in $U$.

\begin{figure}[hbt]
\centering
\begin{tikzpicture}[scale=0.5]
\scope
\clip (0,0) ellipse (5.9cm and 2.5cm);
\fill[rounded corners=0.5cm, pattern=vertical lines,pattern color=red] ( 6.3, - 5.2) rectangle (1.1, 2.6);
\endscope

\scope
\clip (0,-0.55) ellipse (5.1cm and 1.7cm);
\fill[rounded corners=0.5cm, pattern=horizontal lines,pattern color= blue]  ( -6.3, - 5.2) rectangle (-1.1, 2.6);
\endscope

    \draw (0,0) ellipse (5.9cm and 2.5cm);
    \draw (0,-0.55) ellipse (5.1cm and 1.7cm);
    \fill[color = white] (0,-1.1) ellipse (4.1cm and 0.9cm);
    \draw[fill = white, pattern=horizontal lines,pattern color= blue] (0,-1.1) ellipse (4.1cm and 0.9cm);

    \draw[rounded corners=0.5cm, dashed] ( - 6.2, - 4.3) rectangle ( - 1.1, 2.6);
    \draw[rounded corners=0.5cm, dashed]  ( 6.2, -4.3) rectangle (1.1, 2.6);

\node[rectangle, fill = white, fill opacity=0.65] at (0, -1.1) {\color{white} {\SDS}};
\node at (0, -1.1) {{\SDS}};
\node at (0, 0.43) {{\SEDS}};
\node at (0, 1.8) {{\SIM}};

\node[align = center, font = \small] at (-3.6,  -3.3) {$\fixall{\Gamma}$ globally\\ $\ell$-tractable};
\node[align = center, font = \small] at (3.6,  - 3.3) {$\fixall{\Gamma}$ globally\\ $s$-intractable};
\end{tikzpicture}
\hskip0.7cm
\begin{tikzpicture}[scale=0.5]
\scope
\clip (0,0) ellipse (5.9cm and 2.5cm);
\fill[pattern=vertical lines,pattern color=red] ( 7, - 5.2) rectangle (0.5, 2.8);
\endscope

\scope
\clip (0,-0.55) ellipse (5.1cm and 1.7cm); 
\fill[pattern=horizontal lines,pattern color= blue]  ( -7, - 5.2) rectangle (0.5, 2.8);
\endscope

    \draw (0,0) ellipse (5.9cm and 2.5cm);
    \draw (0,-0.55) ellipse (5.1cm and 1.7cm);
    \fill[color = white] (0,-1.1) ellipse (4.1cm and 0.9cm);
    \draw[fill = white, pattern=horizontal lines,pattern color= blue] (0,-1.1) ellipse (4.1cm and 0.9cm);

    \draw[dashed]  (0.5, - 4.3) -- (0.5, 2.8);

\node[rectangle, fill = white, fill opacity=0.65] at ( - 0.6, -1.1) {\color{white} {\SDS}};
\node at ( - 0.6, -1.1) {{\SDS}};
\node[rectangle, fill = white, fill opacity=0.65] at ( - 0.6, 0.43) {\color{white} {\SEDS}};
\node at ( - 0.6, 0.43) {{\SEDS}};
\node at ( - 0.6, 1.8) {{\SIM}};

\node[align = center, font = \small] at (-2.5,  -3.4) {$\fixall{\Gamma}$ globally\\ $\ell$-tractable};
\node[align = center, font = \small] at (3.2,  - 3.4) {$\fixall{\Gamma}$ globally\\ $\ell$-intractable};

\end{tikzpicture}
\caption{Classification of {\SEDS} languages on arbitrary finite domains (left) and on three-element domains (right). A language $\Gamma$ is globally $\ell$-tractable when marked by
horizontal (blue) lines and globally $s$-intractable when marked by vertical (red)
lines, depending on the language $\fixall{\Gamma}$ on a smaller domain. 
(Recall that global $s$-intractability implies global $\ell$-intractability.)
In case of three-element domains, the Boolean language
$\fixall{\Gamma}$ is either globally $\ell$-tractable
or globally $\ell$-intractable, while this is not known for larger domains.
}
\label{figure_three_element_domain}
\end{figure}
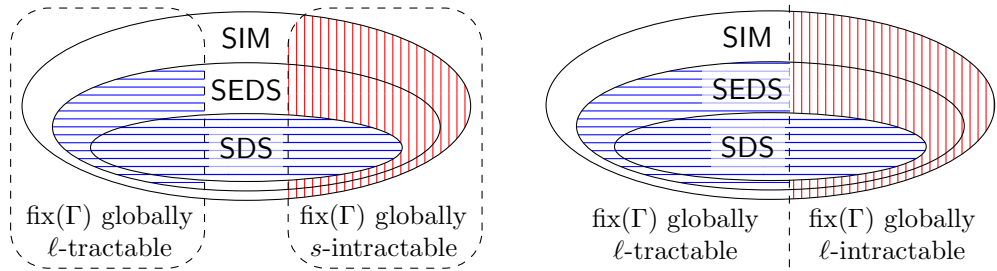

One implication of our results is a dichotomy theorem for lower-bounded VCSPs on the Boolean domain;
every Boolean language is either globally $\ell$-tractable or globally
$\ell$-intractable. Although lower-bounded VCSPs are more general than
surjective VCSPs, this classification coincides with the dichotomy theorem for
surjective VCSPs given by \cite{FUZ18:surjective}. 

In addition, combining our reduction of {\SEDS} languages to a smaller domain and the dichotomy theorem for the Boolean domain leads to a classification of all {\SEDS} languages on three-element domains with respect to $\ell$-tractability, which is featured on the right-hand side of Figure~\ref{figure_three_element_domain}.

The foundation of our results is an extension of the Generalised Min-Cut problem
that might be of independent interest. Given integers $p, q \in \mathbb{N}_{0}$, a graph with non-negative edge weights and a superadditive set
function defined on its vertices, the goal in the Bounded Generalised Min-Cut
problem is, just like in the {\GMC} problem, to find a subset of the vertices
such that the sum of the induced cut and the superadditive set function
evaluated on it are minimal among all possible solutions. The solution space,
however, is restricted to subsets containing at least $q$ and at most all but
$p$ vertices.

If an optimal solution has value $0$, there can be exponentially many optimal
solution, \mbox{e.g.} when there are no edges and the superadditive function always
evaluates to 0. Our main algorithmic result is that, for all instances with non-zero optimal value and for any constant bounds $p, q \in \mathbb{N}_{0}$, all solutions that are optimal up to a constant factor can be
enumerated in polynomial time (and thus, in particular, there are only
polynomially many of them).

We finish with two remarks on, as far as we can tell, unrelated work. First, it
is natural to consider Karger's elegant (randomised) min-cut
algorithm~\cite{Karger:1993}, which also allows to enumerate (polynomially many)
near-optimal cuts, and try to adapt it to the newly introduced Bounded
Generalised Min-Cut problem. Despite trying, we do not see any way of doing
it.\footnote{It appears to be an issue that the superadditive set function is
evaluated only for the solution set, while the set of remaining vertices may
exhibit an excessively large set function value even in an optimal solution.
That makes it implausible to think a local criterion for edge contractions could
incorporate the superadditive set function in a suitable manner, i.e. somehow
preventing the set function value from getting too large.} 
Moreover, we only know how to establish our tractability results on surjective
VCSPs by a reduction to the Bounded Generalised Min-Cut problem that includes
that superadditive function, but that one fails many properties required by
Karger's algorithm. (For instance, superadditive functions are not necessarily
submodular.) Second, it is notationally convenient to go back and forth between
weighted relations (on a domain of size $k + 1$) and $k$-set functions, as we will
explain in Section~\ref{sec:EDS} and use throughout the paper. We do not see a
connection (suggested by an anonymous reviewer of the extended abstract of this
work~\cite{mz19:stacs}) to the characterisation of arc consistency via set
polymorphisms~\cite{Feder:1998,Dalmau99:set}, which are properties of (weighted)
relations but not their equivalent description. More generally, we do not know
whether our tractability result could be established using recent work on
consistency methods for CSPs~\cite{Barto14:jacm} or LP relaxations for
VCSPs~\cite{Kolmogorov2015,tz17:sicomp}.

\paragraph*{Organisation}
We will proceed in the following manner. Section 2 gives a polynomial-time algorithm for enumerating all near-optimal optimal solutions of the Bounded Generalised Min-Cut problem. In Section 3, we extend the notion of {\EDS} to larger domains. A classification of languages from this extension is presented in Section 4. Section 5 provides a dichotomy theorem for lower-bounded VCSPs on the Boolean domain.

\section{The Bounded Generalised Min-Cut Problem}

\label{section:BGMC}

We begin by presenting our algorithm for the Bounded Generalised Min-Cut problem. The problem is based on the notion of superadditive set functions, which we define first.

\begin{definition}
	A \textit{set function} on a finite set $V$ is a function $f : 2^V \rightarrow \overline{\mathbb{Q}}$ defined on subsets of $V$; it is \textit{normalised} if it satisfies $f\left( \emptyset \right) = 0$ and $f\left(X\right) \geq 0$ for all $X \subseteq V$.

	A set function $f$ on $V$ is \textit{increasing} if it is normalised and $f\left(X\right) \leq f\left(Y\right)$ for all $X \subseteq Y \subseteq V$. It is \textit{superadditive} if it is normalised and, for all disjoint $X, Y\subseteq V$, it holds that
\begin{equation}\label{superadditive_equation}
	f\left(X\right) + f\left(Y\right) \leq f\left(X \cup Y\right).\tag{SUP}
\end{equation}
\end{definition}

Since equation (\ref{superadditive_equation}) implies that $f\left(X\right) \leq f\left(X\right) + f\left(Y \backslash X\right) \leq f\left(Y\right)$ for all $X \subseteq Y \subseteq V$, every superadditive set function must also be increasing.

\begin{definition}\label{Definition_Bounded_Generalised_Min_Cut}
	For $p, q \in \mathbb{N}_0$, the \textit{Bounded Generalised Min-Cut} problem with lower bound $q$ and upper bound $p$ is denoted by \BGMC{q}{p}.

	A \BGMC{q}{p} \textit{instance} $h$ is given by an undirected graph $G = \left(V, E\right)$ with edge weights $w : E \rightarrow \mathbb{Q}_{ \geq 0} \cup \left\{ \infty \right\}$ and an oracle defining a superadditive set function $f$ on $V$. For $X \subseteq V$, let $w\left(X\right)  =   \sum_{\substack{\left|\left\{u, v\right\} \cap X\right| = 1}}w\left(\left\{u, v\right\}\right)$ denote the weight of the cut induced by $X$. 

A \textit{solution} of instance $h$ is any set $X \subseteq V$ such that $\left|X\right| \geq q$ and $\left|X\right|  \leq \left|V\right| - p$. The objective is to minimise the value $h\left(X\right) = f\left(X\right) + w\left(X\right)$.
A solution $X$ is \textit{optimal} if the value $h\left(X\right)$ is minimal among all solutions for this instance. We denote the value of an optimal solution by $\lambda$. For any $\alpha  \geq 1$, a solution $X$ is $\alpha$-\textit{optimal} if $h\left(X\right) \leq  \alpha  \lambda$.
\end{definition}

The Generalised Min-Cut problem, simply denoted by {\GMC}, is the Bounded
Generalised Min-Cut problem with lower and upper bound 1. All $\alpha$-optimal solutions of a {\GMC} instance can be enumerated in polynomial time according to \cite[Theorem~5.11]{FUZ18:surjective}, which we restate here.

\begin{theorem}[\cite{FUZ18:surjective}]\label{GMC_Proof} 
	For any instance $h$ of the {\GMC} problem on $n$ vertices with optimal value $0 <  \lambda  <  \infty$ and any constant $\alpha  \in \mathbb{N}$, the number of $\alpha$-optimal solutions is at most $n^{20 \alpha  - 15}$. There is an algorithm that finds all of them in polynomial time.
\end{theorem}

We will assume that all edges are strictly positive-valued, as they can be ignored otherwise. Similarly to \cite[Lemma~53]{FUZ18:surjective} for the {\GMC} problem, we can easily detect and solve the problem when $\lambda  = 0$ or $\lambda  =  \infty$.

\begin{lemma}\label{GMC_Lemma_Value_Distinction}
	For any $p, q \in \mathbb{N}_0$, where $q$ is a constant, a polynomial-time algorithm can determine whether the optimal value of a \BGMC{q}{p} instance $h$ on a graph $G = \left(V, E\right)$ is $\lambda  = 0$, $1 <  \lambda  <  \infty$ or $\lambda  =  \infty$. In case $\lambda  = 0$, it can provide an optimal solution.
\end{lemma}
\begin{proof}
	First, we assume $\lambda  = 0$. Consider some optimal solution $X \subseteq V$. Then $h(X) = 0$ implies that $X$ cannot cut any edges and, hence, must be a union of connected components $C_1, \dots, C_k \subseteq V$ from $G$ for some $k \in \mathbb{N}$. The union $Y = C_1 \cup \dots \cup C_{\min\left(k, q\right)}$ of up to $q$ of those components must still satisfy $q \leq \left|Y\right| \leq  \left|X\right|\leq \left|V\right| - p$ and $h(Y) = 0$, because the superadditive set function $f$ is increasing and $Y\subseteq X$.
	Consequently, an algorithm can check all $O\left(n^q\right)$ combinations of up to $q$ components from $G$ in order to find the solution $Y$. And vice versa, if no such solution of value $0$ is found, it can be concluded that $\lambda > 0$.

	Similarly, to probe whether $\lambda  =  \infty$, we consider those vertices that are connected by infinite-weight edges as components, because any finite-valued solution cannot cut those edges. It is then sufficient to check all $O\left(n^q\right)$ combinations of up to $q$ components to see whether a finite-valued solution exists. Otherwise, if all these candidates have infinite value when they are comprised of $q$ or more vertices, any solution that does not cut any infinity-edges must be a superset of one of these candidates and therefore have infinite value as well due to the increasing nature of $f$.
\end{proof}

Consequently, our goal is to provide a polynomial-time algorithm for enumerating
near-optimal solutions in the case that the optimal value is both positive and
finite. Before doing so, we give two auxiliary lemmas based on \cite[Lemma~5.6]{FUZ18:surjective} and \cite[Lemma~5.10]{FUZ18:surjective}.

\begin{lemma}\label{GMC_Lemma_Induced_Subgraph}
	For any $p, q  \in \mathbb{N}_0$, any \BGMC{q}{p} instance $h$ on a graph $G = \left(V, E\right)$ and any subset $V' \subseteq V$, there is a \BGMC{q}{p} instance $h'$ on the induced subgraph $G\left[V'\right]$ that preserves the objective value of all solutions $X \subseteq V'$. In particular, any $\alpha$-optimal solution $X$ of $h$ such that $X \subseteq V'$ is $\alpha$-optimal for $h'$ as well.
\end{lemma}
\begin{proof}
	Edges with exactly one endpoint in $V'$ need to be taken into account separately because they do not appear in the induced subgraph. We accomplish that by defining the new set function $f'$ by
\begin{equation*}
	f'\left(X\right) = f\left(X\right) +  \sum_{u \in X} \sum_{v \in V \backslash V'}w\left(u, v\right)
\end{equation*}
for all $X \subseteq V'$. By the construction, $f'$ is superadditive, and the objective value $h'\left(X\right)$ for any solution $X \subseteq V'$ equals $h\left(X\right)$.

Note that the minimum objective value for $h'$ is greater than or equal to the minimum objective value for $h$. Therefore, any solution $X \subseteq V'$ that is $\alpha$-optimal for $h$ is also $\alpha$-optimal for $h'$.
\end{proof}

When a solution of some bounded GMC instance is split into two parts, the next lemma gives a bound on the values of these parts based on edges involved in the split.

\begin{lemma}\label{cutsplitting2}
	Let $h$ be a \BGMC{q}{p} instance over vertices $V$ with optimal value $\lambda$ and let $X, Y \subseteq V$ such that $h\left(X\right) \leq  \alpha  \lambda$ and $w\left(Y\right) \leq  \beta  \lambda$ for some $\alpha \geq 1$ and $\beta  \geq 0$. Then it holds
\begin{equation*}
	h\left(X \backslash Y\right) + h\left(X \cap Y\right) \leq \left( \alpha  + 2 \beta \right)\lambda.
\end{equation*}
\end{lemma}
\begin{proof}
	It is well-known and can easily be verified that the cut function $w$ is posimodular, meaning that $w\left(A\right) + w\left(B\right) \geq w\left(A \backslash B\right) + w\left(B \backslash A\right)$ for all $A, B \subseteq V$.
	
	As a consequence, we have
	\begin{align*}
		w\left(X\right) + w\left(Y\right) &  \geq w\left(X \backslash Y\right) + w\left(Y \backslash X\right) \\
		w\left(Y\right) + w\left(Y \backslash X\right) &  \geq w\left(X \cap Y\right) + w\left( \emptyset \right), 
	\end{align*}
	and hence,
	\begin{equation*}
		w\left(X\right) + 2w\left(Y\right) \geq w\left(X \backslash Y\right) + w\left(X \cap Y\right).
	\end{equation*}
	By superadditivity of $f$, it holds $f\left(X\right) \geq f\left(X \backslash Y\right) + f\left(X \cap Y\right)$. The claim then follows from the fact that $f\left(X\right) + w\left(X\right) + 2w\left(Y\right) \leq \left( \alpha  + 2 \beta \right) \lambda $.
\end{proof}

With these preparations on hand, we now proceed with our main algorithmic result. 

\begin{theorem}\label{BGMC_Enumeration_Lower_Bound}
	For some constant $q \geq 2$, let $h$ be a \BGMC{q}{1} instance on a graph $G = \left(V, E\right)$ of size $n = \left|V\right|$ with optimal value $0 <  \lambda  <  \infty$. Let $Y \cup Z = V$ be a partition of $V$ and let $Y_1 \cup \dots \cup Y_k = Y$ for some $k \in \mathbb{N}_0$ be a partition of $Y$ satisfying $0 < \left|Y_i\right| < q$ and $h\left(Y_i\right) \leq \frac{ \lambda }{3q}$ for all $1 \leq i \leq k$.

Then for every constant $\alpha \geq 1$, at most $\frac{\left|Z\right|}{n} \cdot n^{ \tau \left(q, \alpha \right)}$ $\alpha$-optimal solutions $X \subseteq V$ of $h$ satisfy $\left|X \cap Y\right| < q$, where $ \tau \left(q,  \alpha \right) = 60q\alpha  + 41q + 7$. These solutions can all be enumerated in polynomial time.
\end{theorem}
Note that with $Y =  \emptyset$ and $Z = V$, this theorem states for any \BGMC{q}{1} instance  that the number of $\alpha$-optimal solutions is bounded by $n^{\tau\left(q,  \alpha \right)}$.
\begin{proof}
	Proof by induction over $n + \frac{\left|Z\right|}{n + 1}$; that is, induction  primarily over $n$ and, for equal values of $n$, also over $\left|Z\right|$. For $n \leq q$ or $Z =  \emptyset$, there are no solutions of the described form and hence, the statement holds.
	
	Now, fix some $n > q$, some \BGMC{q}{1} instance $h$ on a graph $G = \left(V, E\right)$  of size $n$ with optimum value $0 <  \lambda  <  \infty$ and partitions $Y \cup Z = V$ and $Y_1 \cup \dots \cup Y_k = Y$ as described. By the induction hypothesis, we can assume that the theorem holds for every graph of size $n' < n$ as well as for every partition $\tilde{Y} \cup \tilde{Z} = V$ of graph $G$ satisfying $\left|\tilde{Z}\right| < \left|Z\right|$.
	
	To simplify matters, we can replace any infinite edge weights in $G$ with a
  large value ($\alpha \cdot (1 + f(V) + \sum_{w(u,v) < \infty}w\left(u,v\right))$ works) without affecting any of our assumptions or the set of $\alpha$-optimal solutions we are looking for. Thus, we will subsequently assume that all edges are finite-valued.

	According to Lemma \ref{GMC_Lemma_Induced_Subgraph}, there exists a \BGMC{q}{1} instance $h_Z$ on the induced subgraph $G\left[Z\right]$ that preserves the objective value of every solution $X' \subsetneq Z$ with respect to $h$. In the following, we treat $h_Z$ as a {\GMC} instance (i.e. with lower bound $1$). Let $\lambda_Z$ denote the optimal value of
$h_Z$. We can assume $\lambda_Z <  \infty$ because otherwise, due to the absence of infinite-weight edges in $G$ and the superadditivity of $f$, no finite-valued solution $X \subseteq V$ of $h$ can satisfy $X \cap Z \neq  \emptyset$. Let $Y_{k + 1} \subsetneq Z$ be an optimal solution of $h_Z$, i.e. $h_Z\left(Y_{k + 1}\right) =  \lambda_Z$.

	If $h\left(Y_{k + 1}\right)$ is sufficiently large, we show that it is essentially sufficient to enumerate {\GMC} solutions of $G\left[Z\right]$ up to a constant factor. For small $h\left(Y_{k + 1}\right)$, our strategy will be to reduce the problem to the partition $Y' \cup Z' = V$, where $Y' = Y_1 \cup \dots \cup Y_k \cup Y_{k + 1}$ and $Z' = Z \backslash Y_{k + 1}$. This approach is outlined in Figure \ref{fig:BGMCproof}.

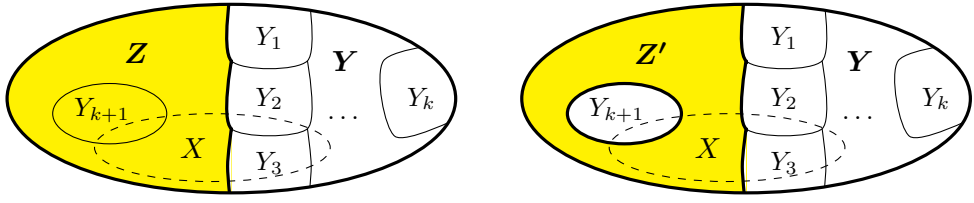
\begin{figure}[hbt]
	\centering
	\begin{tikzpicture}[scale=0.5]

	\scope
	\clip (0,0) ellipse (5.9cm and 2.5cm);
	\fill[color= yellow]  ( -7, - 3) rectangle (0, 3);
	\fill [white] plot [smooth] coordinates {(0,3) (0,1) (2,1) (2,3)};
	\fill [white] plot [smooth] coordinates {(0,1) (0,-0.8) (2, -0.8) (2,1)};
	\fill [white] plot [smooth] coordinates {(0,-0.8) (-0.1, -1.8) (0, -3)};
	
	\endscope
	
	\draw [very thick] (0,0) ellipse (5.9cm and 2.5cm);
	\scope
	\clip (0,0) ellipse (5.9cm and 2.5cm);
	\draw plot [smooth] coordinates {(0,3) (0,1) (2,1) (2,3)};
	\draw plot [smooth] coordinates {(0,1) (0,-0.8) (2, -0.8) (2,1)};
	\draw plot [smooth] coordinates {(0,-0.8) (-0.1, -1.8) (0, -3)};
	\draw plot [smooth] coordinates {(2,-0.8) (2.1, -1.8) (2, -3)};
	\draw plot [smooth] coordinates {(6,1.5) (4,1) (4.3,-1) (6, - 0.5)};
	
	\draw (-3.2 , -0.4) ellipse (1.5cm and 0.8 cm);
	
	\scope
	\clip (-1 , -3) rectangle (0.03, 3);
	\draw [very thick] plot [smooth] coordinates {(0,3) (0,1) (2,1) (2,3)};
	\draw [very thick] plot [smooth] coordinates {(0,1) (0,-0.8) (2, -0.8) (2,1)};
	\draw [very thick] plot [smooth] coordinates {(0,-0.8) (-0.1, -1.8) (0, -3)};
	\endscope
	
	\draw[dashed]  (-0.5,-1.3) ellipse (3.1cm and 0.9 cm);

	\node[align = center, font = \small] at (1, 1.6) {$Y_1$};
	\node[align = center, font = \small] at (1, 0) {$Y_2$};
	\node[align = center, font = \small] at (1, -1.7) {$Y_3$};
	\node[align = center, font = \small] at (3, -0.5) {$\dots$};
	\node[align = center, font = \small] at (5, 0) {$Y_k$};
	\node[align = center] at (3, 1) {$\boldsymbol{Y}$};
	
	\node[align = center] at (-2.5, 1.2) {$\boldsymbol{Z}$};
	\node[align = center, font = \small] at (-3.4,-0.3) {$Y_{k + 1}$};
	
	\node[align = center] at (-1.0, - 1.3) {$X$};
	\endscope

	\end{tikzpicture}
	\hskip0.7cm
	\begin{tikzpicture}[scale=0.5]
	
	\scope
	\clip (0,0) ellipse (5.9cm and 2.5cm);
	\fill[color= yellow]  ( -7, - 3) rectangle (0, 3);
	\fill [white] plot [smooth] coordinates {(0,3) (0,1) (2,1) (2,3)};
	\fill [white] plot [smooth] coordinates {(0,1) (0,-0.8) (2, -0.8) (2,1)};
	\fill [white] plot [smooth] coordinates {(0,-0.8) (-0.1, -1.8) (0, -3)};
	
	\endscope
	
	\draw [very thick] (0,0) ellipse (5.9cm and 2.5cm);
	\scope
	\clip (0,0) ellipse (5.9cm and 2.5cm);
	\draw plot [smooth] coordinates {(0,3) (0,1) (2,1) (2,3)};
	\draw plot [smooth] coordinates {(0,1) (0,-0.8) (2, -0.8) (2,1)};
	\draw plot [smooth] coordinates {(0,-0.8) (-0.1, -1.8) (0, -3)};
	\draw plot [smooth] coordinates {(2,-0.8) (2.1, -1.8) (2, -3)};
	\draw plot [smooth] coordinates {(6,1.5) (4,1) (4.3,-1) (6, - 0.5)};
	
	\draw [very thick, fill= white] (-3.2 , -0.4) ellipse (1.5cm and 0.8 cm);
	
	\scope
	\clip (-1 , -3) rectangle (0.03, 3);
	\draw [very thick] plot [smooth] coordinates {(0,3) (0,1) (2,1) (2,3)};
	\draw [very thick] plot [smooth] coordinates {(0,1) (0,-0.8) (2, -0.8) (2,1)};
	\draw [very thick] plot [smooth] coordinates {(0,-0.8) (-0.1, -1.8) (0, -3)};
	\endscope
	
	\draw[dashed]  (-0.5,-1.3) ellipse (3.1cm and 0.9 cm);

	\node[align = center, font = \small] at (1, 1.6) {$Y_1$};
	\node[align = center, font = \small] at (1, 0) {$Y_2$};
	\node[align = center, font = \small] at (1, -1.7) {$Y_3$};
	\node[align = center, font = \small] at (3, -0.5) {$\dots$};
	\node[align = center, font = \small] at (5, 0) {$Y_k$};
	\node[align = center] at (3, 1) {$\boldsymbol{Y}$};
	
	\node[align = center] at (-2.5, 1.2) {$\boldsymbol{Z'}$};
	\node[align = center, font = \small] at (-3.4,-0.3) {$Y_{k + 1}$};
	
	\node[align = center] at (-1.0, - 1.3) {$X$};
	\endscope

	\end{tikzpicture}
	\caption{Given a partition $V = Y \cup Z$ with $Y = Y_1 \cup \dots \cup Y_k$ of the vertices of a \BGMC{q}{1} instance $h$, we want to find every solution $X$ such that $h\left(X\right) \leq  \alpha \lambda $ and $\left|X \cap Y\right| < q$. Consider the {\GMC} instance $h_Z$ on $G\left[Z\right]$ with optimal solution $Y_{k + 1}$. If $h\left(Y_{k + 1}\right) \geq \frac{ \lambda }{3q}$, $X \cap Z$ must be a near-optimal solution of $h$ (left, Case 1). Otherwise, we apply the induction hypothesis either on the subgraph $G\left[Z'\right]$, where $Z' = Z \backslash Y_{k + 1}$ (right, Case 2a), or on the partition $V = Z' \cup \left(Y_1 \cup \dots \cup Y_{k + 1}\right)$ (Case 2b).}
	\label{fig:BGMCproof}
\end{figure}

	Consider any $\alpha$-optimal solution $X \subseteq V$ of $h$ satisfying $\left|X \cap Y\right| < q$. For some integer $t$, let $i_1, \dots, i_t$ denote indices such that $X \cap Y = X \cap \left(Y_{i_1} \cup \dots \cup Y_{i_t}\right)$, i.e. such that $X$ has vertices only in $Y_{i_1}, \dots, Y_{i_t}$ and $Z$. Since $\left|X \cap Y\right| < q$, we require that $t < q$. Let $U = Y_{i_1} \cup \dots \cup Y_{i_t}$.
	
	\paragraph*{Case 1:} If $\lambda_Z \geq \frac{\lambda}{3q}$, we aim to bound the value $h\left(X \cap Z\right)$ relative to $\lambda_Z$.
	Since $w\left(Y_i\right) \leq h\left(Y_i\right) \leq \frac{\lambda}{3q} $ for every $1 \leq i \leq k$ by assumption, it must hold that
	\begin{align*}
		w\left(U\right)  & =  \sum_{\left|\left\{u, v\right\} \cap U\right| = 1}w\left(\left\{u, v\right\}\right) \leq \sum_{j = 1}^{t}\left(\sum_{\left|\left\{u, v\right\} \cap Y_{i_j}\right| = 1}w\left(\left\{u, v\right\}\right) \right) \\
		& = \sum_{j = 1}^{t}w\left(Y_{i_j}\right)  \leq t \cdot \frac{\lambda}{3q}  < q \cdot \frac{\lambda}{3q}  = \frac{ \lambda }{3}.
	\end{align*}
	According to Lemma \ref{cutsplitting2} with $\beta  = \frac{1}{3}$, it follows that
	\begin{equation*}
		h\left(X \backslash U\right) + h\left(X \cap U\right) \leq \left( \alpha  + \frac{2}{3}\right) \lambda,  
	\end{equation*}
	and in particular, since $X \cap Z = X \backslash U$, we have
	\begin{equation*}
		h\left(X \cap Z\right) \leq \left( \alpha  + \frac{2}{3}\right) \lambda.
	\end{equation*}
	Assuming $ \lambda_Z \geq \frac{\lambda}{3q}  $, we can limit the value $h\left(X \cap Z\right)$ relative to $\lambda_Z$ by
	\begin{equation*}
		\left( \alpha  + \frac{2}{3}\right) \lambda  \leq \left( \alpha  + \frac{2}{3}\right) \cdot 3q\lambda_Z = \left(3q\alpha  + 2q\right) \lambda_Z.
	\end{equation*}
	Given that $X \cap Z \neq  \emptyset$, the above equation implies that if $X \cap Z \subsetneq Z$, then $X \cap Z$ is a $\left(3q\alpha  + 2q\right)$-optimal solution of the {\GMC} instance $h_Z$.
	According to Theorem \ref{GMC_Proof}, there are at most
	\begin{equation*}
		n^{20\left\lceil 3q\alpha  + 2q\right\rceil  - 15} \leq n^{20\left(3q\alpha  + 2q + 1\right) - 15} = n^{60q\alpha + 40q + 5}
	\end{equation*}
	$\left(3q \alpha  + 2q\right)$-optimal solutions of {\GMC} instance $h_Z$, which can all be enumerated in polynomial time. Pairing up these choices for $X \cap Z$, in addition to the possibility $X = Z$, with the at most $\sum_{i = 0}^{q - 1}\binom{n}{i}  \leq \sum_{i = 0}^{q - 1}n^i\leq \sum_{i = 0}^{q - 1}\left(\frac{1}{2}\right)^{q - i}n^q \leq n^{q}$ sets of up to $q - 1$ vertices from $Y$ gives at most
	\begin{equation*}
		\left(n^{60q\alpha + 40q + 5} +1\right)  \cdot n^q\leq n^{60q\alpha + 41q + 6} = \frac{1}{n}  \cdot n^{\tau\left(q,  \alpha \right)} \leq \frac{\left|Z\right|}{n} \cdot n^{\tau\left(q,  \alpha \right)}\tag*{(Case 1)}
	\end{equation*} overall choices for $X$ in this case, as required.

	\paragraph*{Case 2a:} Now, let's assume that $\lambda_Z \leq \frac{\lambda}{3q}$ and furthermore that $\left|X \cap Y'\right| \geq q$, where $Y' = Y \cup Y_{k + 1}$. Then it must hold $w\left(Y_{k + 1}\right) \leq  \lambda_Z \leq \frac{\lambda}{3q}$. Let $U' = Y_{i_1} \cup \dots \cup Y_{i_t} \cup Y_{k + 1}$ so that it holds $X \cap Y' \subseteq U'$. Similar to the previous case, we can bound $w\left(U'\right)$ by
	\begin{equation*}
		w\left(U'\right) \leq w\left(Y_{i_1}\right) + \dots + w\left(Y_{i_t}\right) + w\left(Y_{k + 1}\right) \leq \left(t + 1\right) \cdot \frac{ \lambda }{3q}  \leq q \cdot \frac{\lambda}{3q}  = \frac{ \lambda }{3}.
	\end{equation*}
	According to Lemma \ref{cutsplitting2} with $\beta  = \frac{1}{3}$, it must then hold that
	\begin{equation*}
		h\left(X \backslash U'\right) + h\left(X \cap U'\right) \leq \left( \alpha  + \frac{2}{3}\right) \lambda.
	\end{equation*}
	Assuming that $\left|X \cap Y'\right| \geq q$, the set $X \cap U' = X \cap Y'$ is a solution of $h$ and must have value $h\left(X \cap U'\right) \geq  \lambda$. For $Z' = Z \backslash Y_{k + 1}$, it therefore holds that
	\begin{equation*}
		h\left(X \cap Z'\right) = h\left(X \backslash U'\right)\leq \left( \alpha  + \frac{2}{3}\right) \lambda - h\left(X \cap U'\right) \leq \left( \alpha  - \frac{1}{3}\right) \lambda.
	\end{equation*}
	Let $h_{Z'}$ denote the \BGMC{q}{1} instance on the induced subgraph $G\left[Z'\right]$ that preserves the value of $h$ as detailed in Lemma \ref{GMC_Lemma_Induced_Subgraph}.
	Unless $\left|X \cap Z'\right| < q$ or $X \cap Z' = Z'$, the set $X \cap Z'$ is an $\left( \alpha  - \frac{1}{3}\right)$-optimal solution of $h_{Z'}$ (in particular, this case can be ignored when $\alpha < \frac{4}{3}$).
	By applying the induction hypothesis on $h_{Z'}$ with the trivial partition $ \emptyset  \cup Z' = Z'$, it follows that the number of $\left( \alpha  - \frac{1}{3}\right)$-optimal solutions is at most
	\begin{equation*}
		\frac{\left|Z'\right|}{\left|Z'\right|} \cdot \left(\left|Z'\right|\right)^{ \tau \left(q,  \alpha  - \frac{1}{3}\right)} \leq n^{ \tau \left(q,  \alpha  - \frac{1}{3}\right)}.
	\end{equation*}
	In addition, there are at most $ \sum_{i = 0}^{q - 1}\binom{n}{i}\leq n^{q}$ subsets of $Z'$ that have size less than $q$. Accounting also for the possibility $X \cap Z' = Z'$, there are at most 
	\begin{equation*}
		n^{ \tau \left(q,  \alpha  - \frac{1}{3}\right)} + n^q + 1 \leq 3n^{ \tau \left(q,  \alpha  - \frac{1}{3}\right)} \leq n^{ \tau \left(q,  \alpha - \frac{1}{3}\right) + 1}
	\end{equation*}
	choices for $X \cap Z'$ in this case.
	
	Next, we limit the number of choices for $X \cap Y'$. Since $X$ contains at most $q - 1$ vertices from $Y$ (less than $n^q$ choices) and since $Y_{k + 1}$ contains at most $q - 1$ vertices (at most $2^{q - 1}$ choices), the number of possible choices for $X \cap Y'$ is limited by
	\begin{equation*}
		n^q \cdot 2^{q - 1} \leq n^{2q}.
	\end{equation*}
	
	Pairing up each possible choice for $X \cap Z'$ with each choice for $X \cap Y'$ gives a total of at most
	\begin{equation*}
		n^{\tau\left(q,  \alpha  -\frac{1}{3}\right) + 1} \cdot n^{2q}  = n^{\tau\left(q, \alpha   -\frac{1}{3}\right) + 2q + 1} \leq \frac{1}{n} \cdot n^{\tau\left(q,  \alpha \right)}\tag*{(Case 2a)}
	\end{equation*}
	solutions, where the last inequality follows from the fact that
	\begin{equation*}
		\tau\left(q,  \alpha \right) - \tau\left(q,  \alpha  - \frac{1}{3}\right) = 60q \cdot \frac{1}{3} \geq 2q + 2.
	\end{equation*}

	\paragraph*{Case 2b:} Finally, let's assume that $\lambda_Z \leq \frac{ \lambda }{3q}$ and that $\left|X \cap Y'\right| < q$. Since $h_Z\left(Y_{k + 1}\right) = \lambda_Z <  \lambda$ implies $\left|Y_{k + 1}\right| < q$, we can apply the induction hypothesis for instance $h$ with the partition $Y' \cup Z' = V$ to limit the number of choices for $X$. Consequently, this number is at most
	\begin{equation*}
		\frac{\left|Z'\right|}{n} \cdot n^{\tau\left(q, \alpha \right)} \leq \frac{\left|Z\right| - 1}{n} \cdot n^{\tau\left(q, \alpha \right)}.\tag*{(Case 2b)}
	\end{equation*}

	Summing up the bounds for Case 2a and Case 2b, the overall number of choices for $X$ if $\lambda_Z \leq \frac{ \lambda }{3q}$ is bounded by
	\begin{equation*}
		\frac{1}{n} \cdot n^{\tau\left(q, \alpha \right)} + \frac{\left|Z\right| - 1}{n} \cdot n^{\tau\left(q, \alpha \right)} = \frac{\left|Z\right|}{n} \cdot n^{\tau\left(q, \alpha \right)}.
	\end{equation*}
	This proves the upper bound of $\frac{\left|Z\right|}{n} \cdot n^{\tau\left(q, \alpha \right)}$ solutions of the described form.
	
	A polynomial-time algorithm to enumerate all such solutions follows almost immediately from these calculations. 
	Given that $\lambda$ might not be known beforehand, we simply check both Case 1 and Case 2.

	Note that the induction hypothesis is used only in Case 2a, where all $\left( \alpha  - \frac{1}{3}\right)$-optimal solutions of \BGMC{q}{1} instance $h_{Z'}$ with partition $\emptyset  \cup Z'$ need to be computed, and in Case 2b, where all $\alpha$-optimal solutions of \BGMC{q}{1} instance $h_Z$ with partition $Y' \cup Z'$ need to be computed. 
	It is straightforward to verify that the algorithmic complexity of all required operations except for these two recursive calls can be bounded by some polynomial $\mbox{poly}\left(n\right)$.
	We show by induction that $T_{ \alpha }\left(n, Z\right) = 3 \alpha n^{3 \alpha } \cdot \left|Z\right| \cdot \mbox{poly}\left(n\right)$ is an upper bound on the overall complexity. 
	\begin{align*}
	T_{ \alpha }\left(n, Z\right) &  \leq \mbox{poly}\left(n\right) + T_{ \alpha  - \frac{1}{3}}\left(\left|Z'\right|, Z'\right) + T_{ \alpha }\left(n, Z'\right) \\
		  &  \leq  \mbox{poly}\left(n\right) + \left(3\alpha - 1\right) n^{3\alpha} \cdot \mbox{poly}\left(n\right) + 3 \alpha n^{3 \alpha } \cdot\left(\left|Z\right| - 1\right)\cdot \mbox{poly}\left(n\right) \\
		&  \leq 3 \alpha n^{3 \alpha } \cdot\left|Z\right|\cdot \mbox{poly}\left(n\right)\qedhere
	\end{align*}
\end{proof}

\begin{corollary}\label{BGMC-Enumeration}
	For any $p, q \in \mathbb{N}_0$ and $\alpha \geq 1$, where $q$ and $\alpha$ are constants, and for any \BGMC{q}{p} instance $h$ with optimal value $0 <  \lambda  <  \infty$, all $\alpha$-optimal solutions can be enumerated in polynomial time.
\end{corollary}

\begin{proof}
	Let $h = f + w$ be a \BGMC{q}{p} instance with $0 <  \lambda  <  \infty$. First, we assume that $p \geq 1$ and $q \geq 2$. The superadditive set function
\begin{equation*}
	f'\left(X\right) = \begin{cases} \infty & \text{if }\left|X\right| > \left|V\right| - p \\
	f\left(X\right) & \text{otherwise}\end{cases}
\end{equation*}
defines a \BGMC{q}{1} instance $h' = f' + w$ where every solution $X \subseteq V$ of size $\left|X\right| > \left|V\right| - p$ has infinite value so that the set of finite-valued solutions and their values are identical for $h$ and $h'$. Therefore, it is sufficient to enumerate all $\alpha$-optimal solutions of $h'$, which can be accomplished in polynomial time according to Theorem  \ref{BGMC_Enumeration_Lower_Bound}

If $p = 0$ or $q < 2$, there are up to $\left|V\right| + 2$ additional solutions that can all be checked in polynomial time.
\end{proof}

\section{Extending {\EDS} to Larger Domains}
\label{sec:EDS}

In this section, we formally introduce the classes {\SIM}, {\SEDS} and {\SDS}.
In order to simplify our notation, we will subsequently always consider the $\left(k + 1\right)$-element domain $D = \left\{0, 1, \dots, k\right\}$ for some integer $k$. Any other domain of size $k + 1$ can simply be relabelled without affecting its properties.
One label from the domain will play a special role; without loss of generality
(due to relabellings), it will be $0$.

\subsection{$k$-Set Functions}
It will be convenient to go back and forth between weighted relations and $k$-set functions, which is, subject to a minor technical assumption, always possible.

\begin{definition}
	Let $k \in \mathbb{N}$ and let $V$ be a finite set. A \textit{$k$-set function} on $V$ is a function $f : \left(k + 1\right)^V \rightarrow \overline{\mathbb{Q}}$ defined on $k$-tuples of pairwise disjoint subsets of $V$. 
A $k$-set function $f$ over $V$ is \textit{normalised} if it satisfies $f\left( \emptyset, \dots,  \emptyset \right) = 0$ and $f\left(X_1, \dots, X_k\right) \geq 0$ for all disjoint $X_1, \dots, X_k \subseteq V$.
\end{definition}

Note that a $1$-set function is simply a \textit{set function} as defined in Section \ref{section:BGMC}. 
The correspondence between weighted relations and $k$-set functions is formalised by the next definition.

\begin{definition}
	Let $\gamma$ be an $n$-ary weighted relation on the $\left(k + 1\right)$-element domain $D = \left\{0, 1, \dots, k\right\}$, and let $f$ be the $k$-set function on $V = \left[n\right]$ that is defined for disjoint sets $X_1,\dots, X_{k} \subseteq V$ by $f \left(X_1, \dots, X_{k}\right) = \gamma \left(\boldsymbol{x}\right)$, where the $i$-th coordinate of $\boldsymbol{x}$ is given by $x_i = d$ if $i \in X_d$ for some $0 \neq d \in D$ and $x_i = 0$ otherwise. Then $\gamma$ \textit{corresponds} to $f$.

	Furthermore, we say that $\gamma$ \textit{corresponds under normalisation} to a $k$-set function if $\gamma \left(\boldsymbol{0}^{n}\right) <  \infty$ and $ \gamma \left(\boldsymbol{0}^n\right) \leq \gamma \left(\boldsymbol{x}\right)$ for all $\boldsymbol{x} \in D^n$. In this case, the $k$-set function corresponding under normalisation to $\gamma$ is the normalised $k$-set function corresponding to $\gamma  -  \gamma \left(\boldsymbol{0}^{n}\right)$, i.e. the weighted relation with offset such that the assignment $\boldsymbol{0}^{n}$ evaluates to $0$.
\end{definition}

According to this definition, there is a unique $k$-set function corresponding to every weighted
relation on the $\left(k + 1\right)$-element domain, and vice versa.
Furthermore, assuming that $\gamma \left(\boldsymbol{0}^{n}\right) <  \infty$, a
weighted relation $\gamma$ corresponds under normalisation to a $k$-set function
precisely if it admits multimorphism $\langle c_0 \rangle$, which we will formally define in Section~\ref{sec:Boolean}.

The next definition states when a $k$-set function is approximated by a (1-)set function. This approximation will serve as central tool in order to bring the structure of languages from larger domains essentially down to a Boolean domain.
\begin{definition}\label{Definition-alpha-approximation}
Let $f$ be a $k$-set function and $g$ a set function on $V$. We say that $g$ $\alpha$-\textit{approximates} $f$ if, for all disjoint $X_1,\dots, X_{k} \subseteq V$, it holds that
\begin{equation*}
	g\left(X_1 \cup \dots \cup  X_{k}\right) \leq f\left(X_1,\dots, X_{k}\right) \leq  \alpha  \cdot g\left(X_1  \cup \dots \cup X_{k}\right).
\end{equation*}
\end{definition}

\subsection{Fixing a Label: Reduced Languages}
\label{subsec:fix}

Reducing a language to a smaller domain by fixing the occurrences of label 0, as defined subsequently, will become a central tool in our classification.

\begin{definition}\label{Definition-Fixing-Relation}
	Let $\gamma$ be a weighted relation on domain $D$ of arity $n$ and let $U \subseteq \left[n\right]$. Then $\fix{U}{\gamma}$ is the weighted relation on domain $D^*=D \backslash \left\{0\right\}$ of arity $m = n - \left|U\right|$ defined for $x_1,\dots, x_m \in D^*$ by
	\begin{equation*}
	\fix{U}{\gamma}\left(x_1, \dots, x_m\right) = \gamma \left(y_1,\dots,y_n\right), \quad \text{where } y_i = \begin{cases} 0 & \text{if }i \in U \\
	 x_{\left|\left[i\right]\setminus U\right|} & \text{otherwise.}\end{cases}
	\end{equation*}
\end{definition}

In other words, $\fix{U}{\gamma}$ takes an assignment from domain $D^*$ to all variables except for those with index in $U$, and evaluates it through $\gamma$ by assigning label $0$ to the remaining variables. In Definition~\ref{def:fixlanguage}, we generalise this concept in order to express the language that is generated by fixing every possible assignment of label 0.

\begin{definition}
	\label{def:fixlanguage}
Let $\Gamma$ be a language on domain $D$.
For any $\gamma  \in  \Gamma$, let $\fixall{ \gamma }$ denote the set $\left\{\fix{U}{ \gamma } : U \subseteq \left[ \operatorname{ar}\left( \gamma\right) \right] \right\}$ generated by fixing any possible subset of variables to label 0. We define the language $\fixall{ \Gamma }$ on domain $D^*=D \backslash \left\{0\right\}$ by
$
	\fixall{ \Gamma } =  \bigcup_{ \gamma  \in  \Gamma }\fixall{ \gamma }.
$
\end{definition}

\subsection{Extending {\EDS} to Larger Domains}

The class {\EDS}, or \textit{essentially a downset}, has been introduced in \cite{FUZ18:surjective} for the Boolean domain.
\begin{definition}\label{Definition-EDS}
	For any $\alpha  \geq 1$, a normalised set function $f$ on $V$ is $\alpha$-{\EDS} if, for all $X, Y \subseteq V$, it holds that
\begin{equation}\label{Definition-EDS_equation}
   f\left(X \backslash Y\right) \leq \alpha  \cdot \left( f \left(X\right) + f \left(Y\right)\right).\tag{\EDS}
\end{equation}

A weighted relation is $\alpha$-{\EDS} if it corresponds under normalisation to a set function that is $\alpha$-{\EDS}. Moreover, a language $\Gamma$ is {\EDS} if there is some $\alpha  \geq 1$ such that every weighted relation $ \gamma   \in  \Gamma$ is $\alpha$-{\EDS}.
\end{definition}

Fulla et al. showed~\cite{FUZ18:surjective} that {\EDS} languages are globally $s$-tractable. We improve upon this result by proving that such languages are in fact globally $\ell$-tractable, and we extend the idea of being essentially a downset to larger domains through the classes {\SIM}, {\SEDS} and {\SDS}.

Intuitively, a language is {\SIM}, or \textit{similar to a Boolean language}, if it can be approximated by a language over the Boolean domain using Definition \ref{Definition-alpha-approximation}. More precisely, for each weighted relation, the value of any two assignments that assign label 0 to the same set of variables must be equal up to a constant factor. This way, when disregarding constant factors, all non-zero labels can be treated as a single one, leading us essentially to the Boolean domain. 

\begin{definition}
	Let $f$ be a normalised $k$-set function on set $V$. For any $\alpha  \geq 1$, $f$ is called $\alpha$-{\SIM} if, for all disjoint $X_1,\dots, X_{k} \subseteq V$ and all disjoint $Y_1,\dots, Y_{k} \subseteq V$ such that $X_1 \cup \dots \cup X_{k} = Y_1 \cup\dots \cup Y_{k}$, it holds that
\begin{equation}\label{Definition-SIM_equation}
	f \left(X_1, \dots, X_{k}\right) \leq  \alpha  \cdot f \left(Y_1, \dots, Y_{k}\right).\tag{\SIM}
\end{equation}

A weighted relation is $\alpha$-{\SIM} if it corresponds under normalisation to a $k$-set function that is $\alpha$-{\SIM}. Moreover, a language $\Gamma$ is {\SIM} if there is some $\alpha  \geq 1$ such that every weighted relation $ \gamma   \in  \Gamma$ is $\alpha$-{\SIM}.
\end{definition}

Note that every normalised set function is $1$-{\SIM}. Hence, {\EDS} is a subclass of {\SIM}.
Going beyond the Boolean domain, the class {\SEDS} of languages \textit{similar to {\EDS}} arises as a natural generalisation of {\EDS}. Intuitively, {\SEDS} contains precisely those languages that can be approximated by {\EDS} languages.

\begin{definition}\label{Definition-SEDS}
For any $\alpha  \geq 1$, a normalised $k$-set function $f$ on $V$ is $ \alpha $-{\SEDS} if it is $\alpha$-{\SIM} and, for all disjoint $X_1, \dots, X_{k} \subseteq V$ and all disjoint $Y_1, \dots, Y_{k} \subseteq V$, it holds that
\begin{equation}\label{Definition-SEDS_equation}
    f\left(X_1 \backslash Y_1,\dots,  X_{k} \backslash Y_{k}\right) \leq \alpha \cdot \left(f\left(X_1, \dots, X_{k}\right) + f\left(Y_1,\dots, Y_{k}\right)\right). \tag{\SEDS}
\end{equation}

A weighted relation is $\alpha$-{\SEDS} if it corresponds under normalisation to a $k$-set function that is $\alpha$-{\SEDS}. Moreover, a language $\Gamma$ is {\SEDS} if there is some $\alpha  \geq 1$ such that every weighted relation $ \gamma   \in  \Gamma$ is $\alpha$-{\SEDS}.
\end{definition}

The class {\SDS}, or \textit{similar to a downset}, imposes a stricter requirement than {\SEDS}. When any arguments of a weighted relation are changed to label $0$, the value should decrease, stay equal or increase by at most a constant factor. Intuitively, weighted relations of these languages can be approximated by increasing set functions.

\begin{definition}\label{Definition-SDS}
For any $\alpha  \geq 1$, a normalised $k$-set function $f$ on $V$ is $ \alpha $-{\SDS} if it is $\alpha$-{\SIM} and in addition, for all disjoint $X_1, \dots, X_{k}, Y_1,\dots, Y_{k} \subseteq V$, it holds that
\begin{equation}\label{Definition-SDS_equation}
	 f\left(X_1, \dots, X_{k}\right) \leq  \alpha  \cdot f\left(X_1 \cup Y_1, \dots, X_{k} \cup Y_{k}\right).\tag{\SDS}
\end{equation}

A weighted relation is $\alpha$-{\SDS} if it corresponds under normalisation to a $k$-set function that is $\alpha$-{\SDS}, and a language $\Gamma$ is {\SDS} if there is some $\alpha  \geq 1$ such that every weighted relation $ \gamma   \in  \Gamma$ is $\alpha$-{\SDS}.
\end{definition}

Note that {\SDS} is a subclass of {\SEDS}. To see this, consider any $\alpha$-{\SDS} $k$-set function $f$ on $V$. Then it holds for all disjoint $X_1, \dots, X_{k} \subseteq V$ and all disjoint $Y_1, \dots, Y_{k} \subseteq V$ that
\begin{equation*}
	f\left(X_1 \backslash Y_1,\dots,  X_{k} \backslash Y_{k}\right) \leq  \alpha  \cdot f\left(X_1, \dots, X_{k}\right)  \leq  \alpha  \cdot \left(f\left(X_1, \dots, X_{k}\right) + f\left(Y_1,\dots, Y_{k}\right)\right),
\end{equation*} 
proving that $f$ is $\alpha$-{\SEDS}.

\section{Classifying {\SEDS} and {\SDS} Languages}
\label{sec:classification}

In this section, we first show that a {\SEDS} language $\Gamma$ is globally $\ell$-tractable if it is {\SDS} or if the reduced language $\fixall{\Gamma}$ is globally $\ell$-tractable. Afterwards, we prove global $s$-intractability of the remaining {\SEDS} languages conditioned on global $s$-intractability of $\fixall{\Gamma}$.

We begin by restating~\cite[Theorem~5.17]{FUZ18:surjective} concerning {\EDS} languages and then devise similar approximations for {\SEDS} and {\SDS} languages.

\begin{theorem}[\cite{FUZ18:surjective}]\label{Theorem_EDS_Approximation_FullaZ17}
	For any $\alpha$-{\EDS} set function $f$ on $V$, there exists a {\GMC} instance $h$ that   $\alpha^{n + 2}\left(n^3 + 2n\right)$-approximates $f$, where $n = \left|V\right|$.
\end{theorem}

\begin{lemma}
	For any $ \alpha $-{\SEDS} $k$-set function $f$ on $V$, there exists an $\alpha$-{\EDS} set function $g$ that $\alpha^2$-approximates $f$.
\label{EDT0-EDS-approximation}
\end{lemma}
\begin{proof}
	We define the set function $g$ on $V$ by
$
	 g\left(X\right) =  \frac{1}{ \alpha }f \left(X,  \emptyset, \dots,  \emptyset \right).
$
Observe that, since $f$ is normalised, it holds $g \left( \emptyset \right) = f\left( \emptyset, \dots,  \emptyset \right) = 0$ and $g\left(X\right)  =  \frac{1}{ \alpha }f\left(X,  \emptyset, \dots,  \emptyset \right)\geq 0$ for every $X \subseteq V$. Thus, $g$ is normalised as well.
In addition, for all $X, Y \subseteq V$, it holds that
\begin{align*}
	 \alpha \cdot \left(g \left(X\right) +g\left(Y\right)\right) &  =f \left(X,  \emptyset, \dots,  \emptyset \right) +  f \left(Y,  \emptyset, \dots,  \emptyset \right) {\geq} \frac{1}{ \alpha } \cdot f\left(X \backslash Y,  \emptyset, \dots,  \emptyset \right) 
	  = g\left(X \backslash Y\right),
\end{align*}
where the second step uses equation (\ref{Definition-SEDS_equation}).
Hence, $g$ is $\alpha$-{\EDS}. 

It remains to show that $g$ $\alpha^2$-approximates $f$. For this purpose, consider any disjoint $X_1,\dots, X_{k} \subseteq V$ and let $X =  \bigcup_{i = 1}^{k}X_i$ denote their union. Since $f$ is $\alpha$-{\SIM}, it holds that
\begin{equation*}
	g\left(X\right)  =  \frac{1}{ \alpha }f\left(X,  \emptyset, \dots,  \emptyset \right) \leq  f\left(X_1, \dots, X_{k}\right)\leq  \alpha  \cdot f \left(X,  \emptyset, \dots,  \emptyset \right) = \alpha^2  \cdot g\left(X\right).\qedhere
\end{equation*}
\end{proof}

By combining Lemma~\ref{EDT0-EDS-approximation} and Theorem~\ref{Theorem_EDS_Approximation_FullaZ17}, we can deduce the following result.

\begin{theorem}
\label{SEDS-Approximation}
	For any $\alpha $-{\SEDS} $k$-set function $f$ on $V$, there exists a {\GMC} instance $h$ that $\alpha^{n + 4}\left(n^3 + 2n\right)$-approximates $f$, where $n = \left|V\right|$.
\end{theorem}
\begin{proof}
	Let $f$ be an $\alpha$-{\SEDS} $k$-set function defined on $V$. According to Lemma \ref{EDT0-EDS-approximation}, there exists  an $\alpha$-{\EDS} set function $g$ that $\alpha^2$-approximates $f$, meaning that, for all disjoint $X_1,\dots, X_{k} \subseteq V$, it holds
	\begin{equation}
		g\left( \bigcup_{i = 1}^{k}X_{k}\right) \leq  f \left(X_1,\dots, X_{k}\right) \leq  \alpha^2  \cdot g\left( \bigcup_{i = 1}^{k}X_{k}\right) .
		\label{EDT0-approximation_EDT0-bound}
	\end{equation}
	According to Theorem \ref{Theorem_EDS_Approximation_FullaZ17}, as an $\alpha$-{\EDS} set function, $g$ is $ \alpha^{n + 2}\left(n^3 + 2n\right)$-approximable by some {\GMC} instance $h$, meaning that, for every $X \subseteq V$, it holds
	\begin{equation}
		h\left(X\right) \leq g \left(X\right) \leq  \alpha^{n + 2}\left(n^3 + 2n\right) \cdot h\left(X\right).
		\label{EDT0-approximation_EDS-bound}
	\end{equation}
	By combining (\ref{EDT0-approximation_EDT0-bound}) and (\ref{EDT0-approximation_EDS-bound}), it follows that, for all disjoint $X_1,\dots, X_{k} \subseteq V$, we have
	\begin{equation*}
		h\left( \bigcup_{i = 1}^{k}X_i\right) \leq f\left(X_1, \dots, X_{k}\right) \leq  \alpha^{n + 4}\left(n^3 + 2n\right) \cdot h\left( \bigcup_{i = 1}^{k}X_i\right),
	\end{equation*}
	proving that $h$ $\alpha^{n + 4}\left(n^3 + 2n\right)$-approximates $f$.
\end{proof}

There is a more restrictive approximation of {\SDS} languages through superadditive set functions, which can be though of as {\GMC} instances without edges.

\begin{theorem}
\label{SDS-Approximation}
	For any $ \alpha $-{\SDS} $k$-set function $f$ on $V$, there exists a superadditive set function $g$ that $n\alpha^{n + 1}$-approximates $f$, where $n = \left|V\right|$.
\end{theorem}
\begin{proof}
	Let the set function $g$ on $V$ be given by
	\begin{equation}
		g\left(X\right) = \frac{\alpha^{\left|X\right| - n - 1}\left|X\right|}{ n } \cdot f\left(X,  \emptyset, \dots,  \emptyset \right).
		\label{superadditivedefinition}
	\end{equation}

	Observe that since $f$ is normalised, it holds $g \left( \emptyset \right) = f\left( \emptyset, \dots,  \emptyset \right) = 0$ and $g\left(X\right) \geq \frac{1}{ n\alpha^{n+1} } \cdot f\left(X,  \emptyset, \dots,  \emptyset \right)\geq 0$ for every $X \subseteq V$. Thus, $g$ is normalised as well.
	Moreover, $g$ is superadditive, because for all disjoint $ \emptyset  \neq X, Y\subseteq V$, it holds that
	\begin{align*}
		g & \left(X\right) + g\left(Y\right)  
		=\frac{ \alpha^{\left|X\right| - n - 1} \left|X\right|}{ n } \cdot f \left(X,  \emptyset, \dots,  \emptyset \right) + \frac{ \alpha^{\left|Y\right| - n - 1} \left|Y\right|}{ n } \cdot f\left(Y,  \emptyset, \dots,  \emptyset \right)   \\
      &\overset{(\ref{Definition-SDS_equation})}\leq \frac{ \alpha^{\left|X\right| - n - 1} \left|X\right|}{ n } \cdot  \alpha  \cdot f \left(X \cup Y,  \emptyset, \dots,  \emptyset \right) + \frac{ \alpha^{\left|Y\right| - n - 1} \left|Y\right|}{ n } \cdot  \alpha  \cdot f \left(X \cup Y,  \emptyset, \dots,  \emptyset \right)   \\
      & \overset{X, Y \neq  \emptyset }\leq \frac{ \alpha^{\left|X\right| + \left|Y\right| - n - 1} \left|X\right|}{ n }\cdot f \left(X \cup Y,  \emptyset, \dots,  \emptyset \right) + \frac{ \alpha^{\left|X\right| + \left|Y\right| - n - 1} \left|Y\right|}{ n }\cdot f\left(X \cup Y,  \emptyset, \dots,  \emptyset \right)   \\
     & \overset{X \cap Y = \emptyset }  = \frac{ \alpha^{\left|X \cup Y\right| - n - 1} \left|X \cup Y\right|}{ n } \cdot f\left(X \cup Y,  \emptyset, \dots,  \emptyset \right)\tag{since $X \cap Y = \emptyset$}  \\
		& = g\left(X \cup Y\right).
	\end{align*}
	
	It remains to show that $g$ $n\alpha^{n + 1}$-approximates $f$. Consider any disjoint $X_1, \dots, X_k \subseteq V$ and let $X = \bigcup_{i = 1}^{k}X_i$. If $X =  \emptyset$, it holds $g\left(X\right) = f\left(X_1, \dots, X_k\right) = 0$. Otherwise, it holds on the one hand that
	\begin{align*}
		g\left(X\right) = \frac{\alpha^{\left|X\right| - n - 1}\left|X\right|}{ n } \cdot   f \left(X,  \emptyset, \dots,  \emptyset \right) \leq \frac{1}{ \alpha }f\left(X,  \emptyset, \dots,  \emptyset \right) \overset{(\ref{Definition-SIM_equation})}\leq f \left(X_1, \dots, X_k\right)
	\end{align*}
	and on the other hand that
	\begin{equation*}
		n\alpha^{n + 1} \cdot g\left(X\right) = \alpha^{\left|X\right|} \cdot  \left|X\right| \cdot   f \left(X,  \emptyset, \dots,  \emptyset \right)  \geq  \alpha  \cdot f \left(X,  \emptyset, \dots,  \emptyset \right) \overset{(\ref{Definition-SIM_equation})}\geq f\left(X_1, \dots, X_k\right).\qedhere
	\end{equation*}
\end{proof}

Based on these approximations, we now show our main tractability theorem, which in places closely follows the
proof of~\cite[Theorem~5.18]{FUZ18:surjective}.

\begin{theorem}\label{Theorem-l-tractability-EDT0}
	Let $\Gamma$ be a {\SEDS} language. Then $\Gamma$ is globally $\ell$-tractable if it is {\SDS} or if the reduced language $\fixall{\Gamma}$ is globally $\ell$-tractable. 
\end{theorem}
\begin{proof}
	Let $\Gamma$ be an {\SEDS} language on domain $D$. Then every weighted relation $\gamma  \in  \Gamma$ corresponds under normalisation to a $k$-set function $f_{ \gamma }$. Furthermore, weighted relations in $\Gamma$ are of bounded arity. If $\Gamma$ is {\SDS}, Theorem \ref{SDS-Approximation} implies that for some $\alpha \in \mathbb{N}$, every such $k$-set function $f_{ \gamma }$ can be $\alpha$-approximated by a superadditive set function $h_{\gamma}$. In the following, we treat $h_{ \gamma }$ as a {\GMC} instance without any edge weights. If $\Gamma$ is not {\SDS}, we can still $\alpha$-approximate every $k$-set function $f_{ \gamma }$ by a {\GMC} instance $h_{\gamma}$ according to Theorem \ref{SEDS-Approximation}, but there is no restriction on the edge weights.

Let $l : D \rightarrow \mathbb{N}_0$ be a fixed lower bound and consider any \VCSPl{ \Gamma } instance $I$ with objective function 
\begin{equation*}
	 \phi_I\left(x_1, \dots, x_n\right) =  \sum_{i = 1}^{t}w_i \cdot  \gamma_i\left(\boldsymbol{x}^i\right).
\end{equation*}
Let $f_I$ be the $k$-set function corresponding under normalisation to the objective function $\phi_I$. We construct a {\GMC} instance $h$ that $\alpha$-approximates $f_I$. 
  
  For $i \in \left[t\right]$, we relabel the vertices of $h_{ \gamma_i}$ to match the variables in the scope $\boldsymbol{x}^i$ of the $i$-th constraint (i.e., vertex $j$ is relabelled to $x_j^i$) and identify vertices in case of repeated variables. As the constraint is weighted by a non-negative factor $w_i$, we also scale the weights of the edges of $h_{ \gamma_i}$ and the superadditive set function by $w_i$ (note that non-negative scaling preserves superadditivity). Instance $h$ is then obtained by adding up {\GMC} instances $h_{ \gamma_i}$ for all $i \in \left[t\right]$. In the following, we treat $h$ as a \BGMC{l^*}{l\left(0\right)} instance, where $l^* =  \sum_{i = 1}^{k}l\left(i\right)$. Note that if $\Gamma$ is {\SDS}, $h$ has zero edge weights.

Let $X_0, \dots, X_{k}$ be a partition of $\left[n\right]$ such that $f_I \left(X_1, \dots, X_{k}\right)$ is minimal among all partitions satisfying $\left|X_{d}\right| \geq l\left(d\right)$ for all $d \in D$. In other words, $X_0, \dots, X_{k}$ corresponds to an optimal assignment for instance $I$. Let $X =  \bigcup_{d = 1}^{k}X_d$ denote all indices with non-zero labels.  
In addition, let $Y \subseteq \left[n\right]$ denote an optimal solution of the \BGMC{l^*}{l\left(0\right)} instance $h$ and let $\lambda  = h\left(Y\right)$.

Since $\left|Y\right| \geq l^*$, there must exist some partition $Y_1, \dots, Y_{k}$ of $Y$ such that $\left|Y_d\right| \geq l\left(d\right)$ for all $1 \leq d \leq k$. Because $h$ $\alpha$-approximates $f_I$, it holds that
\begin{equation*}
	 \lambda \leq h\left(X\right) \leq  f_{I}\left(X_1, \dots, X_{k}\right) \leq f_{I}\left(Y_1, \dots, Y_{k}\right) \leq  \alpha  \cdot h\left(Y\right) =  \alpha  \cdot  \lambda.
\end{equation*}
Hence, $X$ is an $\alpha$-optimal solution of $h$.

According to Lemma \ref{GMC_Lemma_Value_Distinction}, it can be determined in polynomial time whether $\lambda  = 0$, $\lambda =  \infty $ or $0  <   \lambda <  \infty$.
Furthermore, in case $\lambda = 0$, a solution $Z$ such that $h\left(Z\right) = 0$ can be found. We explore this case first. Because $Z$ must have size $\left|Z\right| \geq l^*$ as a solution of \BGMC{l^*}{l\left(0\right)} instance $h$, we can select some partition $Z_1, \dots, Z_{k}$ of $Z$ such that $\left|Z_d\right| \geq l\left(d\right)$ for all $1 \leq d \leq k$. Since $h$ $\alpha$-approximates $f_I$, it must hold $f_I\left(Z_1, \dots, Z_{k}\right) \leq \alpha  \cdot h\left(Z\right) = 0$, meaning that $Z_1, \dots, Z_{k}$ represents an optimal assignment for instance $I$.

If $\lambda  =  \infty$, then obviously there are no feasible assignments.

Otherwise, it holds $0 <  \lambda < \infty$. In this case, we distinguish whether $\Gamma$ is {\SDS} or $\fixall{\Gamma}$ is globally $\ell$-tractable.

First, we assume that $\Gamma$ is {\SDS} and hence, that $h$ has zero edge weights. We claim that the size of $X$ is bounded by a constant. For the sake of contradiction, assume that $\left|X\right| \geq \left( \alpha  + 1\right)l^*$. Then there are disjoint subsets $Z_1, Z_2, \dots, Z_{ \alpha  + 1} \subseteq X$ such that $\left|Z_i\right| \geq l^*$ for all $1 \leq i \leq  \alpha  + 1$. Being a solution of $h$, every $Z_i$ must have value at least $h\left(Z_i\right) \geq  \lambda$. Based on the superadditivity of $h$, we arrive at the contradiction
\begin{equation*}
	\left( \alpha  + 1\right) \cdot  \lambda  \leq  h\left(Z_1\right) + \dots + h\left(Z_{ \alpha  + 1}\right)  \leq h\left(X\right) \leq  \alpha  \cdot \lambda.
\end{equation*}

Thus, it must hold $\left|X\right| < \left( \alpha  + 1\right)l^*$. This leaves less than $O\left(n^{\left( \alpha  + 1\right)l^*}\right)$ possible choices for $X$, each of which admits at most $O\left(k^{\left( \alpha  + 1\right)l^*}\right)$ partitions of the form $X_1 \cup  \dots \cup X_k = X$. By checking all of these, we can retrieve the sets $X_1, \dots, X_k$ in polynomial time.

Now, we assume that $\fixall{\Gamma}$ is globally $\ell$-tractable.
According to Corollary \ref{BGMC-Enumeration}, there are only polynomially many $\alpha$-optimal solutions of $h$, all of which can be computed in polynomial time. $X$ must be among those solutions. By repeating the following process for all of them, we can assume that $X$ is known, and so is $X_0 = \left[n\right] \backslash X$.

Let $D^* = D \backslash \left\{0\right\}$ and let $l_{\restriction_{D^*}} : D^* \rightarrow \mathbb{N}$ denote the restriction of $l$ to $D^*$. We consider the \VCSPx{\mbox{\footnotesize{$\text{\footnotesize{$l$}}_{\text{\fontsize{6.4pt}{5.8pt}\selectfont{$\restriction$}}_{\text{\fontsize{4.5pt}{6pt}\selectfont{$D^*$}}}}$} \hskip-0.2cm}}{\fixall{\Gamma}} instance $I_X = \left(X, D^*, \phi_X\right)$, where objective function $\phi_X$ is constructed from $\phi_I$ by fixing label 0 to the variables in $X_0$. This construction can be realised by replacing every weighted relation $\gamma_i$ from $\phi_I$ with $\fix{U_i}{\gamma_i}$ instead, where $U_i$ are the indices of the variables from $X_0$ in the scope of $\gamma_i$, and the remaining variables in the scope of $\gamma_i$ form the new scope for $\fix{U_i}{\gamma_i}$.
According to this construction, by assigning label $0$ to the variables in $X_0$, every assignment for $I_X$ induces an assignment for $I$ with the same objective value. This includes the assignment for $I_X$ represented by the sets $X_1, \dots, X_k$. Thus, an optimal assignment for $I_X$, which can be obtained efficiently when  $\fixall{\Gamma}$ is globally $\ell$-tractable, induces an optimal assignment for $I$.
\end{proof}

\begin{remark}\label{Remark_SEDS_Enumeration}
	In fact, the algorithm presented in Theorem \ref{Theorem-l-tractability-EDT0} can, for every fixed lower bound $l : D \rightarrow \mathbb{N}_0$ and every \VCSPl{ \Gamma } instance $I$ with optimal value $0 <  \lambda  <  \infty$, enumerate all optimal assignments for $I$ in polynomial time if either
	\begin{enumerate}[(i)]
	\item $\Gamma$ is {\SDS}, or
	\item $\Gamma$ is {\SEDS}, $\fixall{\Gamma}$ is globally $\ell$-tractable and for every \VCSPl{\fixall{\Gamma}} instance with optimal value $0 <  \lambda  <  \infty$, all optimal assignments can be enumerated in polynomial time.
	\end{enumerate}
\end{remark}

To complete our analysis of {\SEDS} languages, we will now focus on the case that a language is not {\SDS} and that $\fixall{\Gamma}$ is globally $s$-intractable. Going even beyond {\SEDS}, our main hardness result is that {\SIM} languages are globally $s$-intractable under those circumstances. 

\begin{theorem}\label{Hardness-Theorem-Similar}
	Let $\Gamma$ be a valued constraint language over domain $D$ that is {\SIM}, but not {\SDS}, and let $\fixall{\Gamma}$ be globally $s$-intractable. Then $\Gamma$ is globally $s$-intractable.
\end{theorem}
\begin{proof}
	Since $\fixall{\Gamma}$ is globally $s$-intractable, the domain $D$ must have at least three elements. Let $\alpha  \geq 1$ be such that $\Gamma$ is $\alpha$-{\SIM}. 
	We show that \VCSPs{\fixall{\Gamma}} is reducible to \VCSPs{ \Gamma }.
	
	For this purpose, let $I = \left(V, D^*, \phi_I\right)$ be any \VCSPs{\fixall{\Gamma}} instance on domain $D^* = D \backslash \left\{0\right\}$ with objective function $ \phi_I\left(\boldsymbol{x}\right) =  \sum_{i = 1}^{t}w_i \gamma_i\left(\boldsymbol{x}_i\right)$. 
	
	Every constraint $\gamma_i$ must be of the form $\gamma_i = \fix{U_i}{\sigma_i}$ for some weighted relation $ \sigma_i \in  \Gamma$ and some set $U_i \subseteq \left[\operatorname{ar}\left( \sigma_i\right)\right]$. Let $ \sigma_i'$ denote the identification of the weighted relation $ \sigma_i$ at the coordinates in $U_i$, i.e. such that $ \sigma_i'\left(\boldsymbol{x}_i, 0\right) = \gamma_i \left(\boldsymbol{x}_i\right)$ for every $\boldsymbol{x}_i \in \left(D^*\right)^{\operatorname{ar}\left( \gamma_i\right)}$. Here and later on in the proof, the notation $\sigma_i'\left(\boldsymbol{x}_i, 0\right)$ is shorthand for $\sigma_i'\left(x_{i,1},\dots, x_{i,\operatorname{ar}\left( \gamma_i\right)}, 0\right)$.	Note that $\sigma_i'$ is expressible over $\Gamma$. We will utilise these relations later in the proof in order to express the objective function $\phi_I$ over $\Gamma$.
	
	Let $\varepsilon > 0$ be a lower bound for the smallest positive difference between the values of any two assignments for instance $I$. In other words, we select $\varepsilon$ sufficiently small so that if the objective value of some assignment is $ \kappa $, then there is no other assignment with objective value in $\left( \kappa  -  \varepsilon,  \kappa \right)$ or $\left( \kappa,  \kappa  +  \varepsilon \right)$. Note that $\varepsilon$ can be calculated efficiently by multiplying the denominators of all values that the constraints can obtain and of all weights that occur in $\phi_I$. 
	
	Similarly, let $\omega$ denote an upper bound for the largest finite value that any assignment for instance $I$ can obtain.
	
	If $\Gamma$ is not {\SDS}, in particularly not $\left(\frac{2\left|V\right|^2 \cdot  \omega }{ \varepsilon } \cdot \alpha^4\right)$-{\SDS}, then there must exist a weighted relation $\gamma  \in  \Gamma$ of some arity $r$ and disjoint $X_1, \dots, X_{k}, Y_1, \dots, Y_{k} \subseteq \left[r\right]$ such that, in violation of equation (\ref{Definition-SDS_equation}), the $k$-set function $f$ corresponding under normalisation to $\gamma$ satisfies
	\begin{equation}\label{Hardness-Lemma-Similar_Pairwise-Set-Condition-Unfinished}
		f \left(X_1, \dots, X_{k}\right) > \frac{2\left|V\right|^2 \cdot  \omega }{ \varepsilon } \cdot \alpha^4  \cdot f \left(X_1 \cup Y_1, \dots, X_{k} \cup Y_{k}\right).
	\end{equation}
	Let $X =  \bigcup_{d = 1}^{k}X_d$ and $Y = \bigcup_{d = 1}^{k}Y_d$. Since $f$ is $\alpha$-{\SIM}, we can transform (\ref{Hardness-Lemma-Similar_Pairwise-Set-Condition-Unfinished}) to
	\begin{equation}\label{Hardness-Lemma-Similar_Pairwise-Set-Condition}
		f\left(X,  \emptyset, \dots,  \emptyset  \right) > \frac{2\left|V\right|^2 \cdot  \omega }{ \varepsilon } \cdot \alpha^2  \cdot f\left(X \cup Y,  \emptyset, \dots,  \emptyset  \right).
	\end{equation}
	
	Without loss of generality, we can assume that $\gamma \left(\boldsymbol{0}^r \right) = 0$ so that $\gamma$ and $f$ are interchangeable. In order to simplify notation, we first define the 3-ary weighted relation $\gamma^*$ by $ \gamma^*\left(x, y, z\right) =  \gamma \left(\boldsymbol{s}_{x, y, z}\right)$ for $x, y, z \in D$, where the $i$-th coordinate $s_i$ of $\boldsymbol{s}_{x, y, z}$ is $s_i = x$ if $i \in X$, $s_i = y$ if $i \in Y$ and $s_i = z$ otherwise.

	According to the (\ref{Hardness-Lemma-Similar_Pairwise-Set-Condition}), it holds that
	\begin{equation*}
		\gamma^* \left(1, 0, 0\right)  > \frac{2\left|V\right|^2 \cdot  \omega }{ \varepsilon }   \cdot \alpha^2  \cdot  \gamma \left(1, 1, 0\right).
	\end{equation*}
	Since $\gamma$ is $\alpha$-{\SIM}, this implies for all $x, y, z \in D^*$ that
	\begin{equation}\label{Hardness-Lemma-Similar_Initial-Condition}
		\gamma^* \left(x, 0, 0\right)  > \frac{2\left|V\right|^2 \cdot  \omega }{ \varepsilon } \cdot  \gamma \left(y, z, 0\right).
	\end{equation}
	
	Finally, let $\nu  > 0$ be a sufficiently large value so that, for all $x, y, z \in D$ such that $\gamma^* \left(x, y, z\right) > 0$, it holds that
	\begin{equation}\label{Hardness-Lemma-Similar_Nu}
		\nu  \cdot  \gamma^*\left(x, y, z\right) >  \omega.
	\end{equation}
	
	Based on these definitions, we can now complete the proof. We distinguish two cases.
	
	\paragraph*{Case 1:} First, assume that $\gamma^* \left(1, 1, 1\right) = 0$.
	
	We construct the \VCSPs{ \Gamma } instance $I' = \left(V \cup \left\{z\right\}, D,  \phi_{I'}\right)$ with objective function
	\begin{equation*}
		\phi_{I'}\left(\boldsymbol{x}, z\right) = \sum_{x, y \in V} \nu   \cdot \gamma^*\left(x, y, y\right) +  \sum_{i = 1}^{t}w_i  \sigma_i'\left(\boldsymbol{x}_i, z\right).
	\end{equation*}
	
	From $\gamma^* \left(1, 1, 1\right) = 0$ and the fact that $\Gamma$ is $\alpha$-{\SIM}, it follows that $\gamma^* \left(x, y, y\right) = 0$ for all $x, y \in D^*$. We focus on assignments for $I'$ of the form $\boldsymbol{x}  \in \left(D^*\right)^{\left|V\right|}$ and $z = 0$. For every such assignment, it must hold
	\begin{align*}
		\phi_{I'}\left(\boldsymbol{x}, z\right) = 0 +  \sum_{i = 1}^{t}w_i  \sigma_i'\left(\boldsymbol{x}_i, z\right) =  \sum_{i = 1}^{t}w_i  \gamma_i\left(\boldsymbol{x}_i\right) =  \phi_I\left(\boldsymbol{x}\right).
	\end{align*}
	
	Hence, every assignment for $I'$ of the form $\boldsymbol{x}  \in \left(D^*\right)^{\left|V\right|}$ and $z = 0$ induces an assignment $\boldsymbol{x}  \in \left(D^*\right)^{\left|V\right|}$ for $I$ with the same objective value, and vice versa. In particular, if $I$ is feasible, then there is an assignment for $I'$ of value at most $\omega$.  To show that an optimal assignment for $I$ can be derived from an optimal assignment for $I'$, it remains to show that every minimal assignment for $I'$ must be of the described form, which we do by showing that every assignment violating this form must have value greater than $\omega$.

	Consider any surjective assignment for $\boldsymbol{x}$ and $z$. Since $\left|D\right| \geq 3$, there must be some variable $x \in V$ such that $x \neq 0$. If there was any $y \in V$ with assigned label $y = 0$, then it would hold $\gamma^*\left(x, y, y\right) > 0$ according to (\ref{Hardness-Lemma-Similar_Initial-Condition}) and therefore
	\begin{equation*}
		\phi_{I'}\left(\boldsymbol{x}, z\right) \geq  \nu  \cdot \gamma^*\left(x, y, y\right) \overset{(\ref{Hardness-Lemma-Similar_Nu})} >  \omega.
	\end{equation*}
	Thus, we can assume $\boldsymbol{x}  \in \left(D^*\right)^{\left|V\right|}$ in every minimal assignment. By the surjectivity of the assignment, that implies $z = 0$ and completes the reduction proof in this case.

	\paragraph*{Case 2:} Now, assume that  $ \gamma^* \left(1, 1, 1\right) > 0$.
	In this case, we construct the \VCSPs{ \Gamma } instance $I^* = \left(V \cup \left\{z\right\}, D, \phi_{I^*}, \right)$ with objective function
	\begin{equation*}
		\phi_{I^*}\left(\boldsymbol{x}, z\right) =  \nu  \cdot  \gamma^*\left(z, z, z\right) +  \sum_{x, y \in V}\frac{ \varepsilon  \cdot  \gamma^*\left(x, y, z\right)}{2\left|V\right|^2 \max\limits_{a, b \in D^*} \gamma^*\left(a, b, 0\right)} +  \sum_{i = 1}^{t}w_i  \sigma_i'\left(\boldsymbol{x}_i, z\right).
	\end{equation*}
	
	An assignment of the form $\boldsymbol{x}  \in \left(D^*\right)^{\left|V\right|}$ and $z = 0$ satisfies on the one hand that
	\begin{align*}
		\phi_{I^*}\left(\boldsymbol{x}, z\right) \leq 0 +  \frac{\varepsilon}{2}  +  \sum_{i = 1}^{t}w_i  \gamma_i\left(\boldsymbol{x}_i\right) =   \frac{\varepsilon}{2}  + \phi_I\left(\boldsymbol{x}\right),
	\end{align*}
	and on the other hand that
	\begin{align*}
		\phi_{I^*}\left(\boldsymbol{x}, z\right)  \geq  \sum_{i = 1}^{t}w_i  \gamma_i\left(\boldsymbol{x}_i\right) = \phi_I\left(\boldsymbol{x}\right).
	\end{align*}
	Hence, an assignment for $I^*$ of the form $\boldsymbol{x}  \in \left(D^*\right)^{\left|V\right|}$ and $z = 0$ induces an assignment $\boldsymbol{x}  \in \left(D^*\right)^{\left|V\right|}$ for $I$ of similar value, and vice versa. It remains to show that every minimal assignment for $I^*$ must be of this form. This completes the reduction proof, because by our choice of $\varepsilon$, a minimal assignment for $I^*$ of this form must then induce a minimal assignment for $I$.
	
	By the assumption $\gamma^* \left(1, 1, 1\right) > 0$ and since $\Gamma$ is  {\SIM}, we have $ \gamma^* \left(z, z, z\right) > 0$ for every $z \in D^*$. Thus, every assignment of the form $\boldsymbol{x}  \in D^{\left|V\right|}$ and $z \in D^*$ must have objective value
	\begin{equation*}
		\phi_{I^*}\left(\boldsymbol{x}, z\right) \geq \nu  \cdot  \gamma^*\left(z, z, z\right)\overset{(\ref{Hardness-Lemma-Similar_Nu})} >  \omega
	\end{equation*}
	and thereby cannot be optimal.
	
	Otherwise, when $z = 0$, there must be some $x \in V$ in every surjective assignment such that $x = 1$. If there was any variable $y \in V$ such that $y = 0$, then, for the summand $\gamma^*\left(x, y, z\right)$ in the second part of $\phi_I^*$, it would hold that
	\begin{align*}
		\gamma^*\left(x, y, z\right) = \gamma^*\left(1, 0, 0\right)  > \frac{2\left|V\right|^2 \cdot  \omega }{ \varepsilon }   \cdot  \max\limits_{0 \neq a, b} \gamma^*\left(a, b, 0\right), 
	\end{align*}
	and hence,
	\begin{align*}
		\phi_{I^*}\left(\boldsymbol{x}, z\right) \geq \sum_{x, y \in V}\frac{ \varepsilon  \cdot  \gamma^*\left(x, y, z\right)}{2\left|V\right|^2 \max\limits_{a, b  \in D^*} \gamma^*\left(a, b, 0\right)} >  \omega.
	\end{align*}
	
	Thus, in addition to $z = 0$, it must also hold $\boldsymbol{x} \in \left(D^*\right)^{\left|V\right|}$ in every minimal assignment.
	This reduces \VCSPs{\fixall{\Gamma}} to \VCSPs{ \Gamma } in this case as well and thereby completes our proof.
\end{proof}

\section{Lower-Bounded VCSPs on the Boolean Domain}
\label{sec:Boolean}
In this final section, we prove our dichotomy theorem for lower-bounded VCSPs on the Boolean domain and, in the end, extend this result to {\SEDS} languages on three-element domains. A classification of Boolean surjective VCSPs has been given by  \cite{FUZ18:surjective} based on polymorphisms and multimorphisms~\cite{Jeavons:1997,cohen06:complexitysoft}, which we introduce in the following.

\begin{definition}
	Let $r$ and $s$ be positive integers and let $\gamma$ be a $r$-ary weighted relation on domain $D$. An operation $o : D^s \rightarrow D$ is a \textit{polymorphism} of $ \gamma $ (and $\gamma$ \textit{admits} polymorphism $o$) if, for all $\boldsymbol{x}_1, \dots, \boldsymbol{x}_s \in D^r$ such that $\gamma \left(\boldsymbol{x}_1\right), \dots,  \gamma \left(\boldsymbol{x}_s\right) <  \infty$, it holds $ \gamma \left(o\left(\boldsymbol{x}_1, \dots, \boldsymbol{x}_s\right)\right) <  \infty $, where $o$ is applied componentwise as
	\begin{equation*}
	o\left(\boldsymbol{x}_1, \dots, \boldsymbol{x}_s\right) = \left(o\left(x_{1, 1}, \dots, x_{s, 1}\right), \dots, o\left(x_{1, r}, \dots, x_{s, r}\right)\right).
	\end{equation*}
	A language $\Gamma$ admits polymorphism $o$ if $o$ is a polymorphism of every $\gamma  \in  \Gamma$.
\end{definition}

\begin{definition}
	Let $r$ and $s$ be positive integers and let  $\gamma$ be a $r$-ary weighted relation on domain $D$. A list $\langle o_1, \dots, o_s\rangle$ of $s$-ary polymorphisms of $\gamma$ is a \textit{multimorphism} of $\gamma$ (and $\gamma$ \textit{admits} multimorphism $\langle o_1, \dots, o_s\rangle$) if, for all $\boldsymbol{x}_1, \dots, \boldsymbol{x}_s \in D^r$, it holds that
	\begin{equation*}
		\sum_{i = 1}^{s} \gamma \left(o_i\left(\boldsymbol{x}_1, \dots, \boldsymbol{x}_s\right)\right) \leq  \sum_{i = 1}^{s} \gamma \left(\boldsymbol{x}_i\right).
	\end{equation*}
	A language $\Gamma$ admits multimorphism $\langle o_1, \dots, o_s\rangle$ if every $\gamma  \in  \Gamma$ admits $\langle o_1, \dots, o_s\rangle$.
\end{definition}

For $d \in D$, the constant operation $c_d : D \rightarrow D$ is defined by $c_d\left(x\right) = d$ for every $x \in D$. According to this definition, a language $\Gamma$ admits multimorphism $\langle c_d \rangle$ for some $d \in D$ if every weighted relation $\gamma  \in  \Gamma$ satisfies $ \gamma \left(d, d,\dots, d\right) \leq \gamma \left(\boldsymbol{x}\right)$ for all $\boldsymbol{x} \in D^{\operatorname{ar}\left( \gamma \right)}$. Such a language is always tractable,
but it may not be $s$-tractable or $\ell$-tractable. Note that the class {\SIM} and all subclasses only contain languages that admit multimorphism $\langle c_0 \rangle$, because this is a requirement for corresponding under normalisation to a $k$-set function.

In addition, the following operations for the Boolean domain $D = \{0,1\}$, which were initially given by \cite{cohen06:complexitysoft}, will be relevant for us.
\begin{itemize}
	\item The binary operation $\min$ ($\max$) returns the smaller (larger) of its two arguments with respect to the order $0 < 1$.
	\item The minority operation $\operatorname{Mn} : D^3 \rightarrow D$ is defined for $x, y \in D$ by $\operatorname{Mn}\left(x, x, y\right) = \operatorname{Mn}\left(x, y, x\right) = \operatorname{Mn}\left(y, x, x\right) = y$.
	\item Similarly, the majority operation $\operatorname{Mj} : D^3 \rightarrow D$ is given for $x, y \in D$ by $\operatorname{Mj}\left(x, x, y\right) = \operatorname{Mj}\left(x, y, x\right) = \operatorname{Mn}\left(y, x, x\right) = x$.
\end{itemize}

Furthermore, given a Boolean language $\Gamma$, let $\neg \left( \Gamma \right)$ denote the language where labels $0$ and $1$ are flipped. This can be seen as relabelling the domain so that VCSPs over $\Gamma$ and over $\neg \left( \Gamma \right)$ have the same complexity.

Based on these operations, \cite[Theorem~3.2]{FUZ18:surjective} gives a classification of Boolean $\overline{\mathbb{Q}}$-valued languages with respect to global $s$-tractability, which we restate here.

\begin{theorem}[\cite{FUZ18:surjective}]\label{Boolean-Rational-Infinity-s-Tractability}
	Let $\Gamma$ be a Boolean language. Then $\Gamma$ is globally $s$-tractable if $\Gamma$ is {\EDS}, if $\neg \left( \Gamma \right)$ is {\EDS} or if $\Gamma$ admits any of the following multimorphisms: $\langle  \min,  \min \rangle$, $\langle  \max,  \max  \rangle, \langle  \min,  \max \rangle$, $\langle\operatorname{Mn}, \operatorname{Mn}, \operatorname{Mn}\rangle, \langle\operatorname{Mj}, \operatorname{Mj}, \operatorname{Mj}\rangle$, $\langle\operatorname{Mj}, \operatorname{Mj}, \operatorname{Mn}\rangle$. Otherwise, $\Gamma$ is globally $s$-intractable.
\end{theorem}

Note that if $\text{P} \neq \text{NP}$, global $s$-tractability and global
$s$-intractability are mutually exclusive. In order to extend Theorem
\ref{Boolean-Rational-Infinity-s-Tractability} to lower-bounded VCSPs, we rely
on the results from Section~\ref{sec:classification} as well as the following two auxiliary lemmas.

\begin{lemma}\label{Lemma_EDS_SEDS_Boolean_Domain}
	Let $\Gamma$ be a Boolean language and let $\alpha  \geq 1$. Then $\Gamma$ is $\alpha$-{\SEDS} if and only if it is $\alpha$-{\EDS}.
\end{lemma}
\begin{proof}
	As a Boolean language, $\Gamma$ is $\alpha$-{\SIM} if every weighted relation $\gamma  \in  \Gamma$ corresponds under normalisation to a set function. This is the case if $\Gamma$ is $\alpha$-{\EDS}.
	
	The remainder of the definitions of {\EDS} and {\SEDS} from pages \pageref{Definition-EDS} and \pageref{Definition-SEDS} are equivalent for the Boolean domain, showing the statement.
\end{proof}

Recall that for a label $d \in D$, the constant relation $\rho_d$ is defined by $\rho_d\left(d\right) = 0$ and $\rho_d\left(x\right) =  \infty$ for $d \neq x \in
D$. Let $\mathcal{C}_{D} = \left\{ \rho_d \mid d \in D\right\}$ denote the set of constant unary relations.

\begin{lemma}\label{Tractability-Constants-Global-l-Tractability}
	Let $\Gamma$ be a language on domain $D$ such that $\Gamma  \cup \mathcal{C}_{D}$ is globally tractable. Then $\Gamma$ is globally $\ell$-tractable.
\end{lemma}
\begin{proof}
	Let $l : D \rightarrow \mathbb{N}_0$ be a fixed lower bound, let $l^* =  \sum_{d \in D}l\left(d\right)$ and consider any \VCSPl{ \Gamma } instance $I = \left(V, D,  \phi_I\right)$. There are only $O\left(\left|V\right|^{l^*}\right)$, i.e. polynomially many, ways to select disjoint sets $V_d \subseteq V$ of size $\left|V_d\right| = l\left(d\right)$ for all $d \in D$. For each such choice, we construct a \VCSPof{\Gamma  \cup \mathcal{C}_{D}} instance $I' = \left(V, D,  \phi_I'\right)$, where $\phi_I'$ is constructed from $\phi_I$ by adding a constraint $\rho_d\left(x\right)$ for every $d \in D$ and every $x \in V_d$. These additional constraints guarantee that only those assignments for $I'$ are feasible that respect lower bound $l$.
	
	Conversely, every assignment for $I$ that respects lower bound $l$ is an assignment for some instance $I'$ constructed from some disjoint sets $V_d \subseteq V$ of the described form. Therefore, an assignment that is minimal among all optimal assignments for instances $I'$ must be an optimal assignment for $I$.
\end{proof}

\begin{theorem}\label{Boolean-Rational-Infinity-l-Tractability} Let $\Gamma$ be a Boolean language. Then $\Gamma$ is globally $\ell$-tractable if and only it is globally $s$-tractable. Otherwise, $\Gamma$ is globally $\ell$-intractable. \end{theorem}
\begin{proof}
	We assume that $\text{P} \neq \text{NP}$, because otherwise every language is globally $\ell$-tractable and the statement trivially holds true. If $\Gamma$ is globally $s$-tractable, it must satisfy at least one of the properties listed in Theorem \ref{Boolean-Rational-Infinity-s-Tractability}.
	
	First, we assume that $\Gamma$ admits any of the multimorphisms $\langle  \min,  \min \rangle$, $\langle  \max,  \max  \rangle$, $\langle  \min,  \max \rangle$, $\langle\operatorname{Mn}, \operatorname{Mn}, \operatorname{Mn}\rangle$, $\langle\operatorname{Mj}, \operatorname{Mj}, \operatorname{Mj}\rangle$ or $\langle\operatorname{Mj}, \operatorname{Mj}, \operatorname{Mn}\rangle$. Then $\Gamma \cup \left\{ \rho_0,  \rho_1\right\}$ must be tractable as well, because the constant relations $\rho_0$ and $ \rho_1$ both admit all of these multimorphisms. This implies global $\ell$-tractability of $\Gamma$ according to Lemma \ref{Tractability-Constants-Global-l-Tractability}.
	
	If $ \Gamma $ is {\EDS}, it must also be {\SEDS} according to Lemma \ref{Lemma_EDS_SEDS_Boolean_Domain}. Furthermore, the reduced language $\fixall{\Gamma}$ is trivial in this case and, in particular, globally $\ell$-tractable. Hence, $ \Gamma $ must be globally $\ell$-tractable by Theorem \ref{Theorem-l-tractability-EDT0}. The same applies if $\neg \left( \Gamma \right)$ is {\EDS}.
	
	Otherwise, $\Gamma$ must be globally $s$-intractable according to Theorem \ref{Boolean-Rational-Infinity-s-Tractability}. That immediately implies global $\ell$-intractability.
\end{proof}

Hence, the classification from Theorem \ref{Boolean-Rational-Infinity-s-Tractability} is also valid for lower-bounded VCSPs.
For $\mathbb{Q}$-valued and $\left\{0,  \infty\right\}$-valued languages, a
tighter classification of Boolean surjective VCSPs is provided in
\cite{FUZ18:surjective}, which can in the same way be lifted to lower-bounded VCSPs by
Theorem \ref{Boolean-Rational-Infinity-l-Tractability}. In particular, a Boolean
$\mathbb{Q}$-valued language $\Gamma$ is globally $\ell$-tractable precisely if it is {\EDS}, if $\neg \left( \Gamma \right)$ is {\EDS} or if $\Gamma$ is submodular.

While our focus so far has been on global $s$-tractability and global $\ell$-tractability, there is a further distinction for infinite languages. A language $\Gamma$ is \textit{tractable} if every finite subset $\Gamma' \subseteq  \Gamma$ is globally tractable, and \textit{intractable} if some finite subset is globally intractable. The terms \textit{$s$-tractability} and \textit{$\ell$-tractability} are defined analogously for surjective and lower-bounded VCSPs.
\cite[Remark~2]{FUZ18:surjective} outlines a dichotomy theorem for Boolean languages with respect to $s$-tractability. We lift this result to lower-bounded VCSPs.

\begin{corollary}
	Let $\Gamma$ be a Boolean language. Then $\Gamma$ is $\ell$-tractable if and only it is $s$-tractable. Otherwise, $\Gamma$ is $\ell$-intractable.
\end{corollary}
\begin{proof}
	If $\Gamma$ is $s$-tractable, every finite subset $\Gamma' \subseteq  \Gamma$ is $s$-tractable. Since $s$-tractability and global $s$-tractability coincide for finite languages, every finite $\Gamma' \subseteq  \Gamma $ must be globally $s$-tractable. By Theorem \ref{Boolean-Rational-Infinity-l-Tractability}, every finite $\Gamma' \subseteq  \Gamma $ is then globally $\ell$-tractable and therefore, $\Gamma$ is $\ell$-tractable.
	
	Otherwise, if $\Gamma$ is not $s$-tractable, there must be some finite subset $\Gamma' \subseteq  \Gamma$ that is not $s$-tractable. In this case, $\Gamma'$ cannot be globally $s$-tractable and must instead be globally $\ell$-intractable by Theorem \ref{Boolean-Rational-Infinity-l-Tractability}. Hence, $\Gamma$ is $\ell$-intractable.
\end{proof}

Moreover, we can now classify lower-bounded VCSPs over {\SEDS} languages on three-element domains.
\begin{theorem}
	Let $\Gamma$ be a {\SEDS} language on domain $D = \left\{0, 1, 2\right\}$. Then $\Gamma$ is globally $\ell$-tractable if it is {\SDS} or if $\fixall{\Gamma}$ is globally $\ell$-tractable, and globally $\ell$-intractable otherwise.
\end{theorem}
\begin{proof}
	If $\Gamma$ is {\SDS} or $\fixall{\Gamma}$ globally $\ell$-tractable, the statement follows from Theorem \ref{Theorem-l-tractability-EDT0}. Otherwise, $\fixall{\Gamma}$ must be globally $s$-intractable by Theorem \ref{Boolean-Rational-Infinity-l-Tractability} and the dichotomy from \cite[Theorem 3.2]{FUZ18:surjective}. Hence, $\Gamma$ is globally $s$-intractable by Theorem \ref{Hardness-Theorem-Similar}, which gives the result.
\end{proof}

\section{Conclusions}

Based on the newly introduced Bounded Generalised Min-Cut problem and its
tractability, which might be of independent interest, we have provided a
conditional complexity classification of surjective and lower-bounded {\SEDS}
VCSPs on non-Boolean domains. Consequently, we obtained a dichotomy theorem with
respect to $\ell$-tractability for Boolean domains as well as, more
interestingly, for {\SEDS} languages on three-element domains.

While our results only pertain to languages that admit multimorphism $\langle
c_d \rangle$ for some label $d$ 
we expect our results and techniques to be useful in identifying new s-tractable and
$\ell$-tractable languages going beyond those admitting $\langle c_d \rangle$.

As mentioned in Section~\ref{sec:intro}, globally tractable languages that
include constant relations are also $s$-tractable. It is easy to show the same
for global $\ell$-tractability.
For example, this
shows that well-studied sources of tractability such as
submodularity~\cite{Schrijver2003combinatorial} and its generalisation
$k$-submodularity~\cite{Huber12:ksub}, which are known to be globally
tractable~\cite{Kolmogorov2015}, are also globally $\ell$-tractable.

What other non-Boolean languages are s-tractable and $\ell$-tractable? Our
results are a first step in this direction. In all cases we encountered global
$s$-(in)tractability coincides with global $\ell$-(in)tractability. We do not
know whether this is true in general.

The natural next step is to consider languages on three-element domains. As is
often the case in the (V)CSP literature, languages on three-element domains are
significantly more complex than Boolean languages; for instance, compare
two-element CSPs~\cite{schaefer78:complexity} and three-element CSPs~\cite{Bulatov06:3-elementJACM}. 
There is an interesting surjective CSP on a three-element domain, known
as the \textit{3-No-Rainbow-Colouring} problem~\cite{Bodirsky12:dam}. The
task is to colour the
vertices of a three-regular hypergraph such that every colour is used at least
once, while no hyperedge attains all three colours. It has recently been shown
that the 3-No-Rainbow-Colouring is NP-hard~\cite{Zhuk20:arxiv}. Consequently, we expect that
it should be possible to classify all three-element surjective CSPs and perhaps even all
three-element surjective VCSPs.

\section*{Acknowledgements}

We would like to thank the anonymous referees of both the conference~\cite{mz19:stacs}
and this full version of the paper. We also thank Costin-Andrei Oncescu for
detailed feedback on a previous version of this paper.


\newcommand{\noopsort}[1]{}

\end{document}